\documentclass[12pt,a4paper]{iopart}
\makeatletter
\expandafter\let\csname equation*\endcsname\relax
\expandafter\let\csname endequation*\endcsname\relax
\makeatother
\usepackage[T1]{fontenc}
\usepackage{amsmath,amssymb,amsthm}
\usepackage{microtype}
\usepackage[hidelinks]{hyperref}
\usepackage{graphicx}
\usepackage{placeins}
\usepackage{enumitem}

\newtheorem{theorem}{Theorem}[section]
\newtheorem{lemma}[theorem]{Lemma}
\newtheorem{proposition}[theorem]{Proposition}
\newtheorem{corollary}[theorem]{Corollary}
\newtheorem{definition}[theorem]{Definition}
\newtheorem{remark}[theorem]{Remark}
\theoremstyle{definition}

\theoremstyle{plain}

\newcommand{\R}{\mathbb{R}}
\newcommand{\DKL}{D_{\mathrm{KL}}}
\newcommand{\DivJ}{D_{\!J}}
\newcommand{\TR}{T_{\mathrm{R}}}
\newcommand{\Ebar}{E_{\mathrm{mc}}}
\newcommand{\keywordblock}[1]{\par\medskip\noindent Keywords: #1}

\begin{document}

\title[Finite-state Gibbs from a recognition cost]{A Finite-State Gibbs Construction from a Recognition Cost}
\author{M Simons$^1$ and J Washburn$^1$}
\address{$^1$Recognition Science, Recognition Physics Institute, Austin, Texas, USA}
\ead{msimons@recognitionphysics.org}

\begin{abstract}
On a finite outcome space, the canonical Gibbs distribution is usually obtained by maximizing Shannon entropy at fixed mean of an externally supplied energy functional. This paper studies the finite-state consequences of a ratio-cost construction instead: after adopting the normalized d'Alembert degree-two closure called the Recognition Composition Law (RCL), with unit log-curvature calibration at the reference ratio, the continuous nontrivial positive branch is $J(x)=\tfrac12(x+x^{-1})-1=\cosh(\log x)-1$. Given the induced cost vector $X_\omega=J(r_\omega)$, multinomial counting and convex duality recover the finite-state Gibbs weights and the identity $F_{\mathrm{R}}(q)-F_{\mathrm{R}}(p)=T_{\mathrm{R}}\,D_{\mathrm{KL}}(q\Vert p)$; the entropy-maximization steps are classical once the cost is fixed. New technical content includes a non-asymptotic Stirling bound and soft-shell constrained-type theorems for real-valued costs. A three-state example compares the Gibbs law to squared-log, affinity-as-energy, and Tsallis alternatives at the same cost vector and mean-cost constraint, with sample-size power calculations at fixed RCL ground truth. The framework is conditional on axioms (A1)--(A3) and restricted to finite outcome spaces with strictly positive weights; it does not derive the composition law from a more primitive principle.
\end{abstract}

\ams{82B03, 82B05, 60F10, 94A15}

% Target journal to be chosen after arXiv posting; uncomment and adjust \submitto{...} for IOP submissions.
%\submitto{J. Stat. Mech.}
\maketitle

\keywordblock{classical statistical mechanics, equilibrium and non-equilibrium; entropy; information theory; large deviations in non-equilibrium systems; rigorous results in statistical mechanics.}

\section{Introduction}

On a finite outcome space, adopting the normalized d'Alembert degree-two closure called the Recognition Composition Law (RCL) of Section~\ref{sec:rcl}, together with unit log-curvature at the reference ratio, selects an energy-like cost on the dimensionless ratios that label outcomes. Once this cost is fixed, the standard finite-state machinery of multinomial counting and convex duality gives the Boltzmann--Gibbs measure, the Helmholtz free energy, and the Kullback--Leibler relaxation theorem. The axiomatic commitment is not derived from a more primitive principle; the paper studies its statistical-mechanical consequences.

\subsection{Motivation and scope}
Statistical mechanics rests on the Boltzmann--Gibbs distribution~\cite{Boltzmann1877,Gibbs1902}
\begin{equation}
  p_i=\frac{e^{-E_i/k_{\mathrm{B}}T}}{Z},\qquad
  Z=\sum_i e^{-E_i/k_{\mathrm{B}}T},
  \label{eq:bg}
\end{equation}
where $E_i$ are energy levels, $T$ is temperature, and $k_{\mathrm{B}}$ is Boltzmann's constant. The canonical form~\eqref{eq:bg} is not a postulate: it is derived from more basic postulates---typically the equal \emph{a priori} probability hypothesis, or the microcanonical ensemble in the thermodynamic limit---and the derivation takes the energy functional $E_i$ as given. Rigorous mathematical foundations for this derivation exist~\cite{Ruelle1969}. Jaynes's maximum-entropy program~\cite{Jaynes1957,Jaynes1957II,PresseGhoshLeeDill2013} reframes the Gibbs distribution as the result of maximizing Shannon entropy subject to an energy constraint, but the energy functional is still taken as given. The Shore--Johnson axiomatization~\cite{ShoreJohnson1980,Caticha2012} fixes Shannon entropy as the unique self-consistent inference rule under linear constraints, but again leaves the choice of constraint functional open. Generalizations such as Tsallis non-extensive statistics~\cite{Tsallis1988} and Naudts's $\phi$-exponential families~\cite{Naudts2011} replace the entropy functional itself; for a modern survey of the large-deviation framework that organizes much of this material, see~\cite{Touchette2009}.

On a finite configuration space, the canonical construction can be phrased entirely in terms of constrained empirical compositions: a type plays the role of a macrostate, its multinomial weight determines an entropy density, and the Gibbs measure appears as the maximizer at fixed mean cost. This is the variational picture that underlies equivalence of ensembles and large-deviation rate functions in classical lattice and mean-field models. Here it is specialized to a single affine cost shell on a finite, strictly positive probability space, with a cost $J$ tied to ratio fluctuations rather than to a spatial Hamiltonian.

\paragraph{What is fixed and what is left as input.}
We work on a finite outcome space throughout, and the framework is restricted to that setting; continuum configuration spaces, hard-core or singular interactions, quantum reduced states, spatially extended Gibbs measures, ensemble equivalence in the sense of Ruelle~\cite{Ruelle1969} or Lanford~\cite{Lanford1973}, and the Boltzmann $H$-theorem for many-particle systems lie outside its scope. Within that setting, the framework supplies a single notational thread $X_\omega=J(r_\omega)$ connecting (a)~a multiplicative--reciprocal functional equation on positive ratios, (b)~the resulting closed-form cost, (c)~the multinomial counting that yields Shannon entropy as the rate function, and (d)~the canonical ensemble on the corresponding affine cost shell. The paper's nonstandard input is the cost-determining functional equation and its finite-state consequences; the remaining variational and large-deviation steps are the standard backbone shared with Jaynes's program. Modeling freedom is thereby relocated from the choice of cost functional to the choice of input ratios and a single calibration constant, not eliminated: specifying $\{r_\omega\}$ is itself a modeling commitment, and Section~\ref{sec:other-axiomatics} shows that the choice of the d'Alembert composition law over its Cauchy, Tsallis, or Naudts alternatives is a further modeling commitment with empirically distinguishable consequences (Section~\ref{sec:tsallis-compare}).

Section~\ref{sec:axioms} states the contribution and its expository and methodological components explicitly; Table~\ref{tab:contribution-status} provides a roadmap of where each formal result enters the chain.

\paragraph{What is new here.}
In this sense, the novelty is not a new entropy principle or a new proof of the Gibbs variational principle, but the relocation of the modeling choice to dimensionless ratio data and a reciprocal composition law for their scalar cost. Section~\ref{sec:axioms} and Table~\ref{tab:contribution-status} record the technical components that follow once that commitment is adopted.

\subsection{Recognition Composition Law and fixed cost}
\label{sec:rcl}

Throughout the paper we abbreviate ``Recognition Composition Law'' as RCL: it is the multiplicative--reciprocal functional equation~\eqref{eq:composition} below, which prescribes how the scalar cost $J$ of a positive dimensionless ratio $x$ combines with that of a second ratio $y$ when the two are merged ($xy$) and compared ($x/y$).

\subsubsection{Physical setting: ratios as primitive bookkeeping variables.}
The RCL arises in a ratio-cost framework~\cite{WashburnZlatanovic2026} in which positive dimensionless ratios, rather than energy levels, serve as the primitive bookkeeping variable for the outcomes of a finite system. The natural physical instances are reversible Markov chains and chemical reaction networks: each assigns a positive dimensionless ratio to every microscopic process or state, and each composes those ratios multiplicatively rather than additively.

Two concrete instances illustrate the point. In a reversible continuous-time Markov chain on a state space $\Omega$ with stationary distribution $\pi$ and transition rates $W_{ij}$, attach to each state $a\in\Omega$ a dimensionless ratio $r_a:=\pi_a/\pi_{a_0}$ relative to a chosen reference state $a_0$ (so $r_{a_0}=1$). Detailed balance $\pi_a W_{ab}=\pi_b W_{ba}$ then assigns the edge ratio $r_b/r_a=\pi_b/\pi_a=W_{ab}/W_{ba}$ to the edge $a\!\leftrightarrow\!b$. When two such edges are concatenated, the corresponding odds ratios multiply: the product $r_b/r_a\cdot r_c/r_b=r_c/r_a$ is the natural composition. In a reaction network, equilibrium constants $K_{\mathrm{eq}}=[B]/[A]$ for $A\rightleftharpoons B$ play the same role---combining two reactions $A\rightleftharpoons B$ and $B\rightleftharpoons C$ multiplies their equilibrium constants because affinities $\Delta\mu/k_{\mathrm{B}}T=\log K_{\mathrm{eq}}$ add additively while the dimensionless ratios themselves multiply.
As a reminder, the structural features (S1)--(S3) below are modeling assumptions imported from the bookkeeping conventions used alongside detailed balance and chemical equilibrium, and they are not implied by detailed balance itself.

\paragraph{Modeling conditions on a ratio-cost.}
We adopt three structural features of detailed-balance and equilibrium-constant bookkeeping as modeling assumptions on the scalar cost $J:\R_{>0}\to\R$; they are not implied by detailed balance or any more primitive physical principle, but summarize the symmetries we import from that bookkeeping context:
\begin{enumerate}
\item[(S1)] Reciprocity (orientation invariance). The same edge $a\!\leftrightarrow\!b$ can be read in either direction, so the ratios $r$ and $r^{-1}$ describe the same reversible process; we require $J(x)=J(x^{-1})$.
\item[(S2)] Reference vanishing. The reference ratio $x=1$ corresponds to balanced odds, so $J(1)=0$.
\item[(S3)] Closure under products and quotients. Given two positive ratios $x,y$, the algebraic operations on $\R_{>0}$ produce two further positive ratios---the product $xy$ (path composition) and the quotient $x/y$ (relative odds of two parallel edges). We require that the joint cost $J(xy)+J(x/y)$ be a function only of $J(x)$ and $J(y)$, symmetric under $(x,y)\mapsto(y,x)$ and $(x,y)\mapsto(x^{-1},y^{-1})$.
\end{enumerate}

Conditions (S1)--(S3) do not determine a unique $J$: many functions satisfy them, and (S3) only names a class of composition rules without selecting one. The d'Alembert law~\eqref{eq:composition} below is a separate postulate.

\paragraph{Choice of composition rule.}
We single out the d'Alembert composition law
\[
  J(xy)+J(x/y)=2J(x)J(y)+2J(x)+2J(y),
\]
i.e.\ equation~\eqref{eq:composition} below, by adding one further postulate: $J(xy)+J(x/y)$ depends polynomially on $J(x)$ and $J(y)$ with total degree at most two, and the only constant term is the one forced by (S2). Symmetry of $J(xy)+J(x/y)$ under $(x,y)\mapsto(y,x)$ is automatic from~(S1) and is not a separate postulate. This degree-two closure is a substantive modeling commitment, not a consequence of (S1)--(S3) or of detailed balance alone. The motivation for adopting it is twofold. First, if the composition law is required to be jointly analytic in $J(x)$ and $J(y)$ near the reference state $x=y=1$, the constant, linear, and bilinear terms collected in~\eqref{eq:composition} are exactly the terms that survive after imposing (S1)--(S2). Degree-two closure thus selects the simplest analytic truncation that is compatible with (S1)--(S3) and admits a non-constant solution with $J(1)=0$. Second, the resulting two-variable polynomial dependence is observable: the joint cost $J(xy)+J(x/y)$ depends only on the individual costs $J(x),J(y)$ and their product $J(x)J(y)$, not on any further function of the underlying ratios.

\paragraph{Why degree two rather than higher order.}
The degree-two closure is the minimal analytic closure that includes a genuine pairwise interaction between two ratio-cost observations. A purely additive closure in $J(x)$ and $J(y)$ gives the squared-log branch in log-coordinates; adding the bilinear term $J(x)J(y)$ is the first non-additive correction compatible with a two-ratio composition rule. Operationally, this means that the combined cost of the two observable ratios $xy$ and $x/y$ is determined by $J(x)$, $J(y)$, and their pairwise product, with no auxiliary state, memory variable, or higher-order coupling. The closure is also testable as a finite relation among four observable costs, $J(xy)$, $J(x/y)$, $J(x)$, and $J(y)$. Higher-degree or non-polynomial closures are possible and are not ruled out by (S1)--(S3); the present paper treats degree two as the minimal non-additive analytic closure and studies the consequences of that disciplined modeling choice. A generic degree-three closure compatible with (S1), (S2), $J(1)=0$, and a non-constant continuous solution is tightly constrained, and no closed-form analogue of $\cosh$ is known for the cubic case; that is the operational reason the manuscript restricts to degree two rather than opening a separate classification project.

Among the natural low-degree alternatives considered here are the trivial branch $J\equiv 0$, the additive d'Alembert/Cauchy branch in log-coordinates $G(u+v)+G(u-v)=2G(u)+2G(v)$ (whose continuous solutions $J=c(\log x)^2$ recover the squared-log surrogate revisited in Section~\ref{sec:ancestor}), and the multiplicative Cauchy law $J(xy)=J(x)J(y)$ (incompatible with $J(1)=0$ unless $J\equiv 0$). The list is illustrative rather than exhaustive once the d'Alembert postulate is relaxed; Section~\ref{sec:other-axiomatics} collects a fuller tabulation of alternatives and their maximum-entropy distributions.

In log-coordinates the closure has a transparent form: writing $G(u):=J(e^u)$, the d'Alembert composition becomes the d'Alembert functional equation~\eqref{eq:dAlembert-G}, and Theorem~\ref{thm:wz-uniqueness} below selects $G(u)=\cosh u-1$. The squared-log surrogate $G(u)=u^2/2$ corresponds to the additive Cauchy law and agrees with the d'Alembert solution to second order at $u=0$; Lemma~\ref{lem:ancestor} shows that the d'Alembert solution dominates pointwise everywhere else.

\subsubsection{\texorpdfstring{The composition law and its $y=1$ reduction}{The composition law and its y=1 reduction}.}
The starting point is the d'Alembert-type composition law on $\R_{>0}$ used by Washburn and Zlatanović~\cite{WashburnZlatanovic2026} (their equation~(3)); multiplicative functional equations of this family are classical~\cite{Aczel1966,Kuczma2009}.
The degree-two condition by itself leaves the normalization of the bilinear interaction as a modeling choice; in this paper the RCL is the normalized d'Alembert closure obtained by fixing that coefficient as in equation~\eqref{eq:composition} below, after which the unit log-curvature calibration fixes the remaining scale of the continuous solution.
\begin{equation}
  J(xy)+J(x/y)=2J(x)J(y)+2J(x)+2J(y),
  \label{eq:composition}
\end{equation}
for an unknown cost $J:\R_{>0}\to\R$. The left-hand side aggregates the cost of the merged ratio $xy$ with the cost of the compared ratio $x/y$; the right-hand side combines the individual costs $J(x)$ and $J(y)$ with their multiplicative interaction $J(x)J(y)$. We take~\eqref{eq:composition} as the composition rule adopted in the Recognition framework~\cite{WashburnZlatanovic2026} and study its calibrated consequences on a finite outcome space. An extension of this framework to additional cost structures appears in~\cite{WashburnRahnamai2026}.

Setting $y=1$ in~\eqref{eq:composition} gives
\[
  2J(x) = 2J(x)J(1) + 2J(x) + 2J(1).
\]
Subtracting $2J(x)$ from both sides leaves $0 = 2J(x)J(1) + 2J(1)$, and factoring $2J(1)$ on the right yields
\begin{equation}
  0 \;=\; 2J(x)\,J(1) + 2J(1) \;=\; 2J(1)\bigl(J(x) + 1\bigr),
\label{eq:y1-reduction}
\end{equation}
valid for every $x > 0$. Thus either $J(1) = 0$ or $J(x) = -1$ for every $x>0$ (hence $J\equiv -1$)~\cite{WashburnZlatanovic2026}; in the former case, reciprocity $J(x)=J(x^{-1})$ follows by setting $x=1$ in~\eqref{eq:composition} and using $J(1)=0$~\cite{WashburnZlatanovic2026}. On the nontrivial branch selected below, $G(u)=J(e^u)$ satisfies $G(u)=\cosh u-1$ near $u=0$, hence $G'(0)=0$ and the cost vanishes to first order at the reference ratio.

\subsubsection{Calibration and classification.}
The composition law~\eqref{eq:composition} together with $J(1)=0$ admits the one-parameter family $x\mapsto\cosh(\lambda\log x)-1$ of admissible solutions, indexed by the strength $\lambda\ge 0$ of the cost penalty. To fix this remaining scale, we calibrate $J$ at the reference state by its second derivative in log-coordinates. Define the log-curvature of $J$ at the reference state by
\[
  \kappa(J):=\lim_{t\to 0}\frac{2J(e^t)}{t^2}
\]
whenever the limit exists; equivalently $\kappa(J)=\lim_{x\to 1} 2J(x)/(\log x)^2$. The condition $\kappa(J)=1$ excludes the spurious solution $J\equiv -1$ and fixes the scale in the one-parameter family $x\mapsto \cosh(\lambda\log x)-1$ of solutions of~\eqref{eq:composition} with $J(1)=0$ selected by Theorem~\ref{thm:wz-uniqueness}. The second-order Taylor coefficient is the natural calibration handle because~(S1)--(S2) force the first-order coefficient to vanish. The choice $\kappa(J)=1$ is the natural normalization: it is the unique scale at which $\DivJ(q\Vert p)$ reduces, near $q=p$, to the standard half-Neyman $\chi^2$ statistic $\tfrac12\chi^2(q\Vert p)$ on the simplex (Proposition~\ref{prop:divj-basic}\,(ii)), and equivalently the unique scale at which the squared-log surrogate $\tfrac12(\log x)^2$ from Lemma~\ref{lem:ancestor} agrees with $J$ to second order at $x=1$. Any other choice $\kappa(J)=\lambda^2$ would amount to an overall rescaling of $\beta$ and $\TR$ in Sections~\ref{sec:gibbs}--\ref{sec:fe} without altering the Gibbs measure or the free-energy--Kullback--Leibler identity.

Theorem~\ref{thm:wz-uniqueness} below is the central algebraic input to the rest of the paper. It matches the classification proved in the formal foundations of Washburn and Zlatanović~\cite{WashburnZlatanovic2026} under measurability and a local integrability hypothesis on logarithmic coordinates.

\begin{theorem}[Classification of measurable solutions of the RCL; selection by calibration]
\label{thm:wz-uniqueness}
Let $J:\R_{>0}\to\R$ be Lebesgue measurable and satisfy the composition law~\eqref{eq:composition}. Define $G(u):=J(e^u)$ for $u\in\R$, and assume the \emph{local integrability} condition
\begin{equation}
  \int_{-\varepsilon}^{\varepsilon}\lvert G(u)\rvert\,\mathrm{d}u<\infty
  \qquad\text{for some }\varepsilon>0,
  \label{eq:J-local-int}
\end{equation}
as in Washburn and Zlatanović~\cite{WashburnZlatanovic2026}. Then either $J\equiv -1$, or $J(1)=0$ and there exists $\lambda\ge 0$ with
\[
  J(x)=\cosh(\lambda\log x)-1\quad\text{or}\quad J(x)=\cos(\mu\log x)-1\quad(\mu\ge 0).
\]
Imposing in addition the calibration $\kappa(J)=1$ excludes the spurious branch $J\equiv -1$, the cosine branch (which has $\kappa=-\mu^2\le 0$), and the trivial $\lambda=0$ representative; the remaining solution is $\lambda=1$, giving
\begin{equation}
  J(x)=\tfrac{1}{2}\bigl(x+x^{-1}\bigr)-1.
  \label{eq:J-closed}
\end{equation}
\end{theorem}
\begin{proof}
First set $G(u):=J(e^u)$. With $x=e^u,y=e^v$, equation~\eqref{eq:composition} becomes
\begin{equation}
  G(u+v)+G(u-v)=2G(u)G(v)+2G(u)+2G(v),
  \qquad u,v\in\R,
  \label{eq:dAlembert-G}
\end{equation}
and $G$ is measurable because $J$ is. Setting $u=v=0$ gives $G(0)^2+G(0)=0$, so $G(0)\in\{0,-1\}$. If $G(0)=-1$, then~\eqref{eq:y1-reduction} forces $G(u)\equiv -1$, and $\kappa(J)=\lim_{u\to 0}-2/u^2=-\infty$, contradicting the calibration $\kappa(J)=1$. Hence $G(0)=0$.

Now set $H(u):=G(u)+1$; then $H(0)=1$ and substituting $G=H-1$ into~\eqref{eq:dAlembert-G} gives
\begin{equation}
  H(u+v)+H(u-v)=2H(u)H(v),\qquad u,v\in\R,
  \label{eq:dAlembert-H}
\end{equation}
the classical d'Alembert equation~\cite[Ch.~3, \S2]{Aczel1966}. Under measurability of $H$ (inherited from $J$) together with~\eqref{eq:J-local-int}, Washburn and Zlatanović~\cite{WashburnZlatanovic2026} show that every solution of~\eqref{eq:dAlembert-H} with $H(0)=1$ coincides pointwise with a continuous solution; in particular $H$ takes one of the classical forms
\begin{equation}
  H(u)=\cos(\mu u)\quad\text{or}\quad H(u)=\cosh(\lambda u),
  \qquad\mu,\lambda\ge 0,
  \label{eq:dAlembert-soln}
\end{equation}
as in the continuous classification of Aczél~\cite[Theorem~5, p.~177]{Aczel1966} and Kuczma--Choczewski--Ger~\cite[\S13.1]{Kuczma2009}. (Without such a regularity hypothesis, pathological non-measurable solutions of~\eqref{eq:composition} exist~\cite{WashburnZlatanovic2026}.)

The calibration selects the hyperbolic branch. For the cosine branch, $\kappa(J)=\lim_{u\to 0}2(\cos\mu u-1)/u^2=-\mu^2\le 0$, which cannot equal~$1$. For the hyperbolic branch,
\[
  1=\kappa(J)=\lim_{u\to 0}\frac{2(\cosh\lambda u-1)}{u^2}=\lambda^2,
\]
so $\lambda=1$, $G(u)=\cosh u-1$, and $J(e^u)=\cosh u-1$. Substituting $u=\log x$ and using $\cosh(\log x)=\tfrac12(x+x^{-1})$ gives~\eqref{eq:J-closed}.
\end{proof}

The measurability and local integrability hypotheses align the present statement with the Washburn--Zlatanović formal treatment~\cite{WashburnZlatanovic2026}; they also exclude the non-measurable pathologies recalled there. Continuity is not needed separately: the selected branch is smooth on $\R_{>0}$ once~\eqref{eq:J-closed} is identified.

The function $J$ measures the cost of a ratio $x$ deviating from unity. It satisfies $J(x)\ge 0$ with equality if and only if $x=1$, and $J(x)=J(1/x)$.

\subsection{Assumptions, contributions, and summary of results}
\label{sec:axioms}
The primitive modeling assumptions are:
\begin{enumerate}
\item[(A1)] The composition law~\eqref{eq:composition} for $J:\R_{>0}\to\R$, with motivation discussed in Section~\ref{sec:rcl}.
\item[(A2)] Unit log-curvature: $\kappa(J)=1$.
\item[(A3)] Lebesgue measurability of $J$ on $\R_{>0}$ together with the local integrability condition~\eqref{eq:J-local-int} on $G(u)=J(e^u)$.
\end{enumerate}
Theorem~\ref{thm:wz-uniqueness} shows that, conditional on (A1)--(A3), the selected cost is $J(x)=\tfrac12(x+x^{-1})-1$. The integrability condition in~\eqref{eq:J-local-int} is the same standing hypothesis used in Washburn and Zlatanović~\cite{WashburnZlatanovic2026} to rule out pathological non-measurable solutions while retaining measurability in place of continuity.

\paragraph{What is and what is not derived.} The axioms (A1)--(A3) are not derived from a more primitive principle. (A1) is the d'Alembert composition~\eqref{eq:composition}; we further restrict to it from the broader (S1)--(S3) class by adopting an analytic-degree-2 closure (Section~\ref{sec:rcl}). (A2) fixes a normalization. (A3) is a measurability and local-integrability regularity package matching~\cite{WashburnZlatanovic2026}. Joint adoption of these gives the closed form $J(x)=\tfrac12(x+x^{-1})-1$ and from there the standard exponential-family / Gibbs structure. The framework is conditional on the modeling commitment (A1)--(A3); it does not derive that commitment, and Section~\ref{sec:other-axiomatics} catalogues the alternatives that become admissible if any one of (S1)--(S3) or the degree-2 closure is dropped.

\paragraph{Contributions and expository components.}
The construction keeps the cost-selection, counting, and variational steps explicit. Its components are:
\begin{enumerate}[align=left, widest={(N2-soft, methodological)}, labelsep=0.55em, leftmargin=*, itemsep=2pt]
\item[(N1, expository)] The cost-classification theorem under (A1)--(A3) (Theorem~\ref{thm:wz-uniqueness}), proved in the measurability plus local integrability regime of Washburn and Zlatanović~\cite{WashburnZlatanovic2026}. The proof reduces~\eqref{eq:composition} to the d'Alembert functional equation in log-coordinates~\cite{Aczel1966,Kuczma2009} and invokes their classification of measurable solutions before applying the calibration $\kappa(J)=1$.
\item[(N2, methodological)] We record an explicit non-asymptotic Stirling bound for the multinomial weights (Proposition~\ref{prop:stirling-rate}), with all constants tracked through Robbins' two-sided Stirling inequalities~\cite{Robbins1955}. This is a direct application of a 1955 inequality, not a new estimate; its purpose is to make the regime of validity of the leading-order rate equality of Theorem~\ref{thm:types} (i.e.\ $N\gg K\log N$) checkable at finite $N,K$.
\item[(N2-soft, methodological)] The soft-shell constrained multinomial theorem (Theorem~\ref{thm:types}) for real-valued costs and shrinking tolerance windows, together with its multi-constraint extension (Proposition~\ref{prop:soft-shell-multi}), removes the exact rational-cost restriction needed for the exact-shell corollary (Corollary~\ref{cor:types-exact}) in the experimentally relevant setting in which measured ratios and costs are not rationally engineered.
\item[(N3, expository)] We unify the four steps under a common notation $\{r_\omega\}\to J\to X\to(\text{Gibbs},F_{\mathrm{R}},\DKL)$ and add (i) a tabulation of alternative composition laws and their maximum-entropy distributions (Section~\ref{sec:other-axiomatics}, Table~\ref{tab:composition-comparison}) and (ii) a side-by-side numerical comparison with Tsallis $q_T$-exponentials at the same cost vector and constraint (Section~\ref{sec:tsallis-compare}). The alternative-composition tabulation lets a reader judge which parts of the cost selection come from (S1)--(S3) alone and which parts come from the additional degree-two closure.
\end{enumerate}
The practical role of (N3) is inferential rather than thermodynamic: when the primary microscopic data are dimensionless ratios---odds ratios in a reversible Markov chain estimated from trajectories, equilibrium constants in a chemical network measured spectroscopically, detailed-balance ratios in a single-molecule experiment---the present chain converts those ratios directly into a finite-state canonical ensemble without an intervening Hamiltonian assignment. Modeling freedom is relocated from the choice of energy functional to the choice of ratio assignment together with the d'Alembert composition law (Section~\ref{sec:rcl}); whether that relocation is useful depends on whether ratios or energies are the more directly accessible objects in the application at hand. Sections~\ref{sec:ledger-bridge} and~\ref{sec:scope-and-applications} discuss representative settings, Section~\ref{sec:discriminating-observables} quantifies the experimental sensitivity required to distinguish the RCL prediction from naive log-cost (``affinity-as-energy'') and Tsallis alternatives, and Section~\ref{sec:stoch-thermo} relates the RCL cost $J(r)$ to the antisymmetric affinity $\log r$ of stochastic thermodynamics on the same edge data.

After $J$ is fixed, the remaining statistical-mechanical machinery is standard: the method-of-types and large-deviation steps (Theorem~\ref{thm:types}), the finite-state Gibbs variational principle (Proposition~\ref{thm:maxent}), the free-energy--KL identity (Proposition~\ref{thm:fe-kl} and its corollaries), and detailed-balance KL monotonicity (Proposition~\ref{prop:KL-markov}) are used here in their classical finite-state forms and credited explicitly to~\cite{Ellis1985,Schnakenberg1976,CoverThomas2006,CsiszarKorner1981,DemboZeitouni2010,Touchette2009,LevinPeresWilmer2017}. Lemma~\ref{lem:ancestor} and Proposition~\ref{prop:divj-basic} record elementary consequences of the specific generator $J$. Table~\ref{tab:contribution-status} lists where each result enters the chain.

The reader who is uninterested in the cost-determining functional equation can read Sections~\ref{sec:gibbs}--\ref{sec:fe} as a self-contained treatment of the canonical ensemble on a finite space with an arbitrary fixed cost vector $X_\omega$. The role of Sections~\ref{sec:rcl}--\ref{sec:divj} is to show how, in ratio-bookkeeping settings such as detailed balance and chemical equilibrium, the cost vector is selected once one adopts the multiplicative--reciprocal closure and calibration on the underlying ratios.

\paragraph{Recognition temperature and free energy; positive-temperature scope.}
We refer to the dimensionless quantity $\TR:=1/\beta$, with $\beta$ the Lagrange multiplier for the mean-cost constraint, as the recognition temperature, and to $F_{\mathrm{R}}=\langle X\rangle_q-\TR H(q)$ as the recognition free energy. Both names are by analogy with the canonical ensemble: $\TR$ and $F_{\mathrm{R}}$ are the Legendre-conjugate variables to mean cost and entropy, and obey the standard Helmholtz algebra in the cost notation $X_\omega=J(r_\omega)$. $\TR$ has dimensions of cost, not Kelvin, and the identification $\TR=k_{\mathrm{B}}T$ for a concrete material requires a microscopic model that the present finite-state framework does not supply.\footnote{All quantities in Sections~\ref{sec:gibbs}--\ref{sec:fe} are dimensionless: $r_\omega$ is a positive number, $J(r_\omega)$ is a positive number, and so are $X_\omega$, $\Ebar$, $\beta$, $\TR=1/\beta$, $Z$, $F_{\mathrm{R}}$, and $H$. To map this dimensionless framework onto a model in which $X_\omega$ represents a physical energy $E_\omega$ in joules, fix an energy scale $\varepsilon_0$ (with units of joules) and set $E_\omega:=\varepsilon_0 X_\omega$, $T:=\varepsilon_0\TR/k_{\mathrm{B}}$ (so $k_{\mathrm{B}}T=\varepsilon_0\TR$, a numerical equality between two energies in joules), and $F^{\mathrm{SI}}:=\varepsilon_0 F_{\mathrm{R}}$. The Boltzmann factor $e^{-E_\omega/k_{\mathrm{B}}T}=e^{-\beta X_\omega}$ is invariant; all identities of Sections~\ref{sec:gibbs}--\ref{sec:fe} carry over unchanged after this rescaling. The bridge constant $\varepsilon_0$ must be supplied by a microscopic model and is not part of the present finite-state framework.} Sections~\ref{sec:gibbs}--\ref{sec:fe} restrict attention to $\beta>0$ (equivalently $\Ebar$ below the uniform-cost mean, see Proposition~\ref{prop:inv-beta}); the negative-temperature regime $\beta<0$ is treated alongside the cosine branch of~\eqref{eq:dAlembert-soln} in~\cite{WashburnZlatanovic2026}, with a recent $H$-theorem for bounded-spectrum systems in this regime in~\cite{Lucente2025}, and is not pursued here.

\paragraph{Roadmap.}
Section~\ref{sec:ancestor} establishes the quadratic lower bound on $J$ in log-ratio coordinates. Section~\ref{sec:divj} defines the $J$-cost divergence $\DivJ$ and its half-Neyman form. Section~\ref{sec:counting} recovers Shannon entropy from multinomial counting and proves the soft-shell constrained multinomial theorem (Theorem~\ref{thm:types}), its exact-shell corollary for rational data (Corollary~\ref{cor:types-exact}), and the non-asymptotic Stirling bound of (N2); Section~\ref{sec:extensions} records the multi-constraint soft-shell extension (N2-soft). Section~\ref{sec:gibbs} produces the finite-state Gibbs law via Lagrange optimization. Section~\ref{sec:fe} proves the free-energy--Kullback--Leibler identity and the detailed-balance monotonicity of $\DKL$. Section~\ref{sec:example} works through a three-state example, including a numerical relaxation trajectory under a detailed-balance Markov dynamics, the Tsallis comparison, and the sample-size power calculation at fixed RCL ground truth (Section~\ref{sec:discriminating-observables}). Section~\ref{sec:discussion} addresses context and limitations, including the connection to stochastic thermodynamics (Section~\ref{sec:stoch-thermo}) and to the maximum-entropy and large-deviation literature; Section~\ref{sec:conclusion} closes.

\begin{table}[htbp]
  \centering
  \small
  \setlength{\tabcolsep}{5pt}
  \renewcommand{\arraystretch}{1.15}
  \begin{tabular}{p{0.34\textwidth} p{0.56\textwidth}}
    \hline
    Result & Role in the chain\\\hline
    Theorem~\ref{thm:wz-uniqueness} (classification of $J$ under measurability and local integrability;~\cite{WashburnZlatanovic2026}) & Reduces (A1)--(A3) to the d'Alembert functional equation~\cite{Aczel1966,Kuczma2009} in log-coordinates and selects the hyperbolic branch; supplies the closed cost $J(x)=\tfrac12(x+x^{-1})-1$ used everywhere downstream.\\
    Lemma~\ref{lem:ancestor} (squared-log lower bound on $J$) & Hyperbolic-cosine identity $J(e^t)=\cosh t-1\ge t^2/2$; feeds the squared-log bound on $\DivJ$ and the surrogate comparison in the example of Section~\ref{sec:example}.\\
    Proposition~\ref{prop:divj-basic} ($\DivJ$ properties; half-Neyman form) & $f$-divergence with generator $f=J$; supplies a closed Neyman-$\chi^2$ form on the simplex used in the numerical example.\\
    Theorem~\ref{thm:types}, Corollary~\ref{cor:types-exact}, Propositions~\ref{prop:stirling-rate}, \ref{prop:soft-shell-multi} (constrained types; exact rational shell; non-asymptotic and multi-constraint soft-shell bounds) & Method-of-types statement on the affine cost shell selected by $X_\omega=J(r_\omega)$, with explicit finite-$N,K$ control via Robbins' inequalities; soft-shell and multi-constraint extensions for irrational measured costs.\\
    Proposition~\ref{thm:maxent} (Gibbs variational principle) & Lagrangian construction in the cost notation $X_\omega=J(r_\omega)$; supplies the canonical Gibbs weights and the dual parameter $\beta$.\\
    Proposition~\ref{thm:fe-kl}, Corollary~\ref{cor:fe-min} (free-energy--KL identity) & Algebraic identity $F_{\mathrm{R}}(q)-F_{\mathrm{R}}(p)=\TR\DKL(q\Vert p)$ in the recognition variables, with the corresponding minimum-free-energy dual. Remark~\ref{cor:entropy-max} records that the maximum-entropy dual is logically identical to Proposition~\ref{thm:maxent}.\\
    Proposition~\ref{prop:KL-markov}, Corollary~\ref{cor:monoFE-markov} (monotone $\DKL$ under detailed balance) & Closes the chain at the level of detailed-balance relaxation: $t\mapsto\DKL(q(t)\Vert p)$ is nonincreasing along reversible Markov dynamics, hence so is $F_{\mathrm{R}}(q(t))$.\\
    \hline
  \end{tabular}
  \caption{Where each formal result enters the chain $\{r_\omega\}\to J\to X\to(\text{Gibbs},F_{\mathrm{R}},\DKL)$. Contributions (N1)--(N3) are summarized in Section~\ref{sec:axioms}; (N1) and (N3) are expository, (N2) records Robbins' non-asymptotic bound, and (N2-soft) gives the soft-shell full-sequence theorem in single- and multi-constraint form.}
  \label{tab:contribution-status}
\end{table}
\FloatBarrier

\section{\texorpdfstring{The RCL cost $J$ and its quadratic lower bound}{The RCL cost J and its quadratic lower bound}}
\label{sec:ancestor}

For every $x>0$, the closed-form cost~\eqref{eq:J-closed} is bounded below by its second-order Taylor expansion in log-coordinates: $J(x)\ge (\log x)^2/2$, with equality iff $x=1$ (Lemma~\ref{lem:ancestor}). The proof is a one-line consequence of the power series $\cosh t-1=\sum_{k\ge 1}t^{2k}/(2k)!$. We record the bound because it is applied pointwise to likelihood ratios in the lower bound on $\DivJ$ in Section~\ref{sec:divj}, and because it sets up the squared-log surrogate cost vector $X^{\mathrm{SQ}}_\omega=(\log r_\omega)^2/2$ used in the discriminating-observable analysis of Section~\ref{sec:example}.

\paragraph{Log-coordinate form.}
In logarithmic coordinates the cost simplifies considerably. Writing any ratio $x>0$ as $x=e^t$ with $t=\log x$, the RCL cost becomes $J(e^t)=\cosh t-1$, so multiplicative departures of $x$ from $1$ are represented by additive deviations $t=\log x$. The reciprocal symmetry $J(x)=J(1/x)$ reduces in these coordinates to the evenness $\cosh(-t)=\cosh t$, and the vanishing of $J$ at $x=1$ to $\cosh 0-1=0$.

\paragraph{The quadratic lower bound.}
Since $\cosh t-1=\sum_{k\ge 1} t^{2k}/(2k)!$, every Taylor coefficient at $t=0$ beyond second order is nonnegative, so $\cosh t-1\ge t^2/2$ for all real~$t$, with equality only at $t=0$. Equivalently, $2(\cosh t-1)=(e^{t/2}-e^{-t/2})^2$ reduces the bound to $\lvert\sinh u\rvert\ge\lvert u\rvert$ for $u=t/2$. In the variable $x$ this is the following statement; Section~\ref{sec:divj} applies it pointwise to likelihood ratios when bounding $\DivJ$ (Proposition~\ref{prop:divj-basic}).

\begin{lemma}[Quadratic lower bound on $J$]
\label{lem:ancestor}
For all $x>0$,
\begin{equation}
  J(x)\ge \frac{(\log x)^2}{2},
  \label{eq:ancestor}
\end{equation}
with equality if and only if $x=1$.
\end{lemma}
\begin{proof}
Set $t=\log x$. Then $J(x)=\cosh t-1=\sum_{k\ge 1} t^{2k}/(2k)!\ge t^2/2=(\log x)^2/2$, with equality if and only if $t=0$.
\end{proof}

The bound~\eqref{eq:ancestor} is used twice downstream: pointwise on likelihood ratios in Proposition~\ref{prop:divj-basic}\,(iii) to control $\DivJ$ from below by a $p$-weighted squared logarithm, and at the level of single ratios in Section~\ref{sec:example} to compare the RCL cost vector with the squared-log surrogate $X^{\mathrm{SQ}}_\omega=(\log r_\omega)^2/2$ on the same three-state example (Table~\ref{tab:discriminate}).

Figure~\ref{fig:ancestor} compares the two functions. The RCL cost $J$ and the squared-log surrogate $(\log x)^2/2$ share a common second-order expansion at $x=1$ (so the bound is sharp there: replacing $1/2$ by any $c>1/2$ would violate~\eqref{eq:ancestor} in a neighborhood of $x=1$). Away from $x=1$, the gap is the strictly positive tail $\sum_{k\ge 2}t^{2k}/(2k)!$ of the hyperbolic-cosine series; in particular, no constant $c>1/2$ satisfies $J(x)\ge c(\log x)^2$ for all $x>0$, because $J(e^t)=\cosh t-1$ grows like $\tfrac12 e^{|t|}$ in the log-coordinate $t=\log x$ (equivalently, $J(x)$ grows linearly in $\max(x,1/x)$ as a function of the ratio variable $x$), while $(\log x)^2$ grows only quadratically in $|\log x|$. Thus $J(x)/(\log x)^2\to\infty$ as $x\to\infty$, while $J(x)/x\to 1/2$: $J$ grows faster than any power of $|\log x|$, but only linearly in the ratio variable $x$ itself.

\begin{figure}[htbp]
  \centering
  \includegraphics[width=0.6\textwidth]{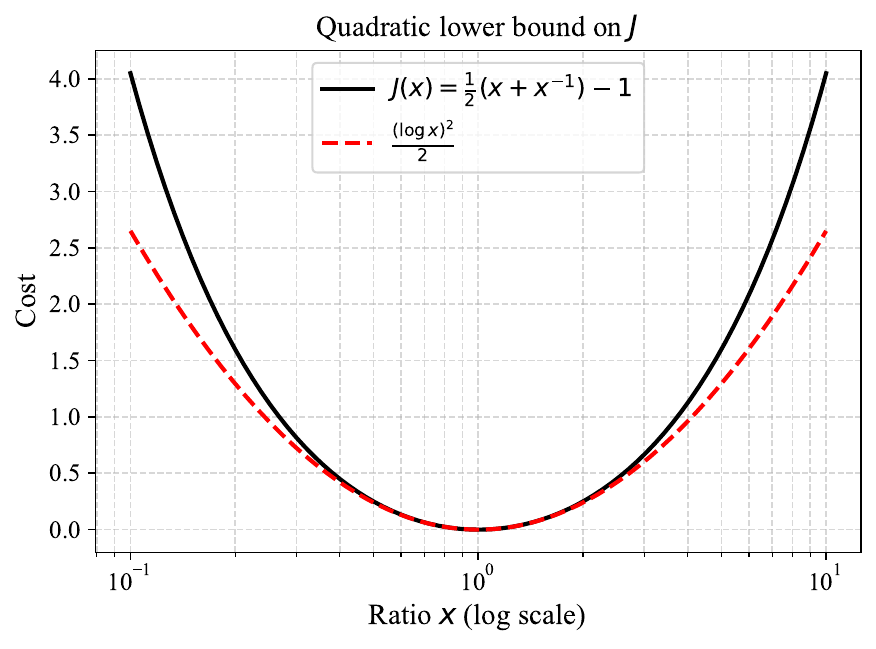}
  \caption{Quadratic lower bound on $J$. The RCL cost $J(x)=\tfrac{1}{2}(x+x^{-1})-1$ (solid curve) and its quadratic surrogate $(\log x)^2/2$ (dashed curve). The two curves share a common second-order expansion at $x=1$; the gap grows for $x\ll 1$ or $x\gg 1$, where $J(e^t)$ grows like $\tfrac12 e^{|t|}$ in the log-coordinate $t=\log x$ (equivalently, $J(x)$ grows linearly in $\max(x,1/x)$), while $(\log x)^2/2$ grows only quadratically in $|t|$.}
  \label{fig:ancestor}
\end{figure}

\paragraph{Interpretation: squared-log fluctuations.}
Lemma~\ref{lem:ancestor} is a pointwise comparison between two functions of a single ratio; it does not identify $J$ with Shannon entropy, Fisher information, or a relative-entropy functional. Shannon entropy enters only later, through the multinomial counting of Section~\ref{sec:counting}; it lives on probability vectors, not on ratios, and is supplied by combinatorics rather than by the functional equation~\eqref{eq:composition}.

Lemma~\ref{lem:ancestor} shows that the RCL-selected cost controls the squared logarithmic deviation of each ratio from the reference state. Two regimes are easy to describe. Near equilibrium ($x\approx 1$), $J(x)\approx (x-1)^2/2\approx (\log x)^2/2$, so any near-equilibrium expansion of a $J$-based quantity reproduces, at second order, the corresponding squared-log expression. Far from equilibrium ($x\gg 1$ or $x\ll 1$), the strict convexity of $J(x)=\tfrac12(x+x^{-1})-1$ penalizes large multiplicative deviations more sharply than the quadratic surrogate, since $J(e^t)$ grows like $\tfrac12 e^{|t|}$ in the log-coordinate $t=\log x$ (equivalently, $J(x)$ grows linearly in $\max(x,1/x)$), whereas $(\log x)^2/2$ grows only quadratically in $|t|$. Thus $J$ sits above a squared-log fluctuation measure in ratio coordinates, while $H$ remains the combinatorial rate functional for empirical compositions, and any far-from-equilibrium dynamics built on $J$ would penalize large multiplicative deviations more sharply than a quadratic log-ratio surrogate.

\section{\texorpdfstring{The $J$-cost divergence}{The J-cost divergence}}
\label{sec:divj}
Throughout this section, $\Omega$ denotes a finite set and $p,q$ are strictly positive probability vectors on~$\Omega$, i.e.\ $p,q:\Omega\to(0,1]$ with $\sum_{\omega\in\Omega} p(\omega)=\sum_{\omega\in\Omega} q(\omega)=1$; in this section $p$ plays the role of a generic reference distribution and $q$ that of a generic competitor, with no Gibbs interpretation. From Section~\ref{sec:gibbs} onwards, $p$ is reserved for the Gibbs reference~\eqref{eq:gibbs} (parametric family $p^{(\beta)}$ when the dependence on $\beta$ is exhibited, and $p^\star=p^{(\beta^\star)}$ when the constraint-matching value $\beta^\star$ from Proposition~\ref{prop:inv-beta} is fixed), while $q$ continues to denote a generic competitor distribution. Strict positivity avoids boundary issues peripheral to the arguments below; the divergence extends to the boundary by lower-semicontinuity, but that extension is not used here. The Kullback--Leibler (relative-entropy) divergence~\cite{KullbackLeibler1951}
\[
  \DKL(q\Vert p):=\sum_{\omega\in\Omega} q(\omega)\log\frac{q(\omega)}{p(\omega)}
\]
will be used at one point in Remark~\ref{rem:divj-vs-kl} for comparison with $\DivJ$; its main role is in Section~\ref{sec:fe}, where it is reintroduced with the same convention $0\log 0:=0$.

The pointwise comparison in Lemma~\ref{lem:ancestor} concerns two functions of a single positive ratio. A probabilistic version requires aggregating across outcomes: at each outcome $\omega$ we form the likelihood ratio $q(\omega)/p(\omega)$, evaluate the RCL cost $J(\cdot)$ at this ratio, and average with respect to the reference distribution~$p$. The result is a divergence in the standard probabilistic sense---a nonnegative functional of the pair $(q,p)$ that vanishes precisely at coincidence---and it is the natural ratio-level companion of the cost fixed in Section~\ref{sec:ancestor}.

\subsection{\texorpdfstring{Definition and half-Neyman $\chi^2$ form}{Definition and half-Neyman chi-squared form}}
Define the $J$-cost divergence by
\begin{equation}
  \DivJ(q\Vert p)=\sum_{\omega\in\Omega} p(\omega)\,J\!\left(\frac{q(\omega)}{p(\omega)}\right).
  \label{eq:DJ-def}
\end{equation}
This also places $J$ near familiar ratio-based contrasts, since the symmetrized Itakura--Saito divergence on a single positive ratio is proportional to $r+r^{-1}-2$.
This is the standard recipe for converting an even, convex, single-variable cost on positive ratios into a divergence on a finite simplex; in the language of information theory, $\DivJ$ is the Csisz\'ar $f$-divergence $D_f(q\Vert p)=\sum_\omega p_\omega f(q_\omega/p_\omega)$ with $f=J$~\cite{Csiszar1967,CsiszarKorner1981,AliSilvey1966}. Here $J$ is not an arbitrary convex $f$ but is fixed by the functional equation~\eqref{eq:composition} together with the unit log-curvature calibration, so the generator in~\eqref{eq:DJ-def} is not a separate modeling choice.

Since $J(x)=\tfrac{1}{2}(x+x^{-1})-1$ is a rational function of $x$, the sum~\eqref{eq:DJ-def} expands explicitly. Substituting $J(q_\omega/p_\omega)=\tfrac12\bigl(q_\omega/p_\omega+p_\omega/q_\omega\bigr)-1=(q_\omega-p_\omega)^2/(2 p_\omega q_\omega)$ into~\eqref{eq:DJ-def} and multiplying by $p_\omega$ produces the equivalent half-Neyman representation
\begin{equation}
  \DivJ(q\Vert p)=\frac{1}{2}\sum_{\omega\in\Omega}\frac{\bigl(p(\omega)-q(\omega)\bigr)^2}{q(\omega)},
  \label{eq:DJ-chisq}
\end{equation}
identifying $\DivJ$ with one-half of Neyman's $\chi^2$ divergence $\sum_\omega (p_\omega-q_\omega)^2/q_\omega$ with arguments oriented as in~\eqref{eq:DJ-def}. Here $q$ appears in the denominator of each summand; this is Neyman's orientation, as opposed to Pearson's $\chi^2$, which places $p$ in the denominator~\cite{Csiszar1967}. The representation makes it manifest that $\DivJ\ge 0$ with equality only at $p\equiv q$; furthermore $J''(x)=x^{-3}>0$ on $\R_{>0}$, so $J$ is strictly convex and the abstract theory of $f$-divergences applies without modification.

\subsection{Positivity and squared-log domination}
Lemma~\ref{lem:ancestor} feeds directly into a lower bound on $\DivJ$. Applying~\eqref{eq:ancestor} pointwise to the likelihood ratio $q_\omega/p_\omega$, multiplying by $p_\omega$, and summing over $\omega$ yields control of $\DivJ$ by a squared logarithmic deviation. The following proposition combines the closed-form representation~\eqref{eq:DJ-chisq} with this bound; both are used in the example of Section~\ref{sec:example}.

\begin{proposition}[Basic properties of $\DivJ$]
\label{prop:divj-basic}
For all positive probability distributions $p,q$ on~$\Omega$:
\begin{enumerate}
\item[(i)] $\DivJ(q\Vert p)\ge 0$, with equality if and only if $q=p$;
\item[(ii)] $\DivJ(q\Vert p)$ is given by~\eqref{eq:DJ-chisq};
\item[(iii)] one has the squared-log lower bound
\begin{equation}
  \DivJ(q\Vert p)\ge \frac{1}{2}\sum_{\omega\in\Omega} p(\omega)\,\bigl(\log(q(\omega)/p(\omega))\bigr)^2.
  \label{eq:divj-dom}
\end{equation}
\end{enumerate}
\end{proposition}
\begin{proof}
For~(i), each summand in $\DivJ$ is nonnegative because $J\ge 0$, and $J(x)=0$ if and only if $x=1$, so the sum vanishes if and only if $q_\omega/p_\omega=1$ for all~$\omega$. Part~(ii) is the algebraic expansion already given in~\eqref{eq:DJ-chisq}. For~(iii), apply Lemma~\ref{lem:ancestor} pointwise to $x=q_\omega/p_\omega$, multiply by $p(\omega)$, and sum over $\omega\in\Omega$.
\end{proof}

The closed form~\eqref{eq:DJ-chisq} is convenient for numerical work---it is what is used in the three-state example of Section~\ref{sec:example}---but $\DivJ$ is not a substitute for the Kullback--Leibler divergence. It does not generate the free-energy identity that organizes Section~\ref{sec:fe}; rather, it is the $f$-divergence naturally associated with the ratio-level cost $J$, and Proposition~\ref{prop:divj-basic}\,(iii) is the divergence-level reflection of Lemma~\ref{lem:ancestor}. The KL geometry that drives the variational principle re-enters at the level of logarithms of probabilities, not at the level of ratios, and Section~\ref{sec:fe} keeps $\DivJ$ and $\DKL$ formally separate for this reason.

The right-hand side of~\eqref{eq:divj-dom} is itself a squared-log quantity. To leading quadratic order near $q=p$, it is one half of the usual Pearson chi-square quadratic form, $\tfrac12\sum_\omega (q_\omega-p_\omega)^2/p_\omega$, but the two expressions are not identical away from coincidence. Proposition~\ref{prop:divj-basic}\,(iii) provides one-sided pointwise control, not a two-sided comparison with $\DKL$.

\subsection{Relation to KL geometry}
Both $\DivJ$ and $\DKL$ are $f$-divergences on the simplex (with $f=J$ and $f(x)=x\log x-x+1$ respectively), both vanish iff $q=p$, and both are jointly convex in $(q,p)$. They are nevertheless not interchangeable as ordering rules.

\begin{remark}[$\DivJ$ is not interchangeable with $\DKL$]
\label{rem:divj-vs-kl}
Distinct $f$-divergences need not induce the same ordering among alternatives to a fixed reference. For $p=(0.3,0.7)$, $q^{(1)}=(0.1,0.9)$, $q^{(2)}=(0.6,0.4)$ on $\lvert\Omega\rvert=2$, one finds $\DKL(q^{(1)}\Vert p)\approx 0.116<\DKL(q^{(2)}\Vert p)\approx 0.192$ but $\DivJ(q^{(1)}\Vert p)\approx 0.222>\DivJ(q^{(2)}\Vert p)\approx 0.188$. Any criterion---gradient flow, model selection, rate-function minimization---that uses one divergence as a contrast may therefore prefer a different alternative under the other.
\end{remark}

The variational principle for the Gibbs distribution (Section~\ref{sec:gibbs}) and the free-energy identity (Section~\ref{sec:fe}) both pass through $\DKL$, since the log-affine form of the Gibbs measure and Shannon entropy on probability vectors couple to log-probabilities, not to ratios; $\DKL$ is the $f$-divergence aligned with that structure on the simplex. From Section~\ref{sec:counting} onward, $\DKL$ is the operative divergence; $\DivJ$ serves only as the ratio-level companion of the cost $J$, where it appears in Proposition~\ref{prop:divj-basic} and in the squared-log comparison of Table~\ref{tab:three}.

To summarize the role of this section: \eqref{eq:DJ-chisq} gives a closed-form numerical expression for $\DivJ$ that is used directly in Section~\ref{sec:example-numerics}, while~\eqref{eq:divj-dom} feeds the squared-log surrogate cost vector $X^{\mathrm{SQ}}_\omega$ used in the discriminating-observable analysis (Table~\ref{tab:discriminate}). The KL divergence $\DKL$ takes over from Section~\ref{sec:counting} onward and is the divergence that organizes the free-energy identity in Section~\ref{sec:fe}.

\section{Finite-state sample model and entropy from counting}
\label{sec:counting}
\paragraph{Notation used in Sections~\ref{sec:counting}--\ref{sec:fe}.}
For convenience we collect the recurring symbols. $\Omega=\{\omega_1,\ldots,\omega_K\}$ is a finite outcome space of size $K\ge 2$. For each outcome $\omega$, the positive ratio $r_\omega>0$ (the input ratio for the modeling problem) determines the cost $X_\omega:=J(r_\omega)\in\R$ via the closed form~\eqref{eq:J-closed}. A probability distribution on $\Omega$ is denoted $p$ when it is the Gibbs reference measure $p^{(\beta)}_\omega\propto e^{-\beta X_\omega}$ and $q$ when it is a generic alternative. The constraint value $\Ebar\in\R$ is the prescribed mean of the cost vector $X$ (target for $\langle X\rangle_p$ under the variational problem); it is not a physical energy, which appears as $E_\omega$ only in the SI bridge in the footnote to the recognition-temperature discussion in Section~\ref{sec:axioms}. The Lagrange multiplier $\beta\in\R$ is conjugate to that mean; $\TR=1/\beta$ for $\beta>0$; $Z(\beta)=\sum_\omega e^{-\beta X_\omega}$ is the partition function. The empirical type $\hat p_k=n_k/N$ is reserved for finite-sample frequency vectors in this section. $H$ denotes Shannon entropy on $\Delta_K$; $\DKL$ and $\DivJ$ are as in Section~\ref{sec:divj}; $F_{\mathrm{R}}=\langle X\rangle-\TR H$ is the recognition free energy. Throughout this section and Sections~\ref{sec:gibbs}--\ref{sec:fe} we assume $X$ is non-constant on $\Omega$ and $\min_\omega X_\omega<\Ebar<\max_\omega X_\omega$; without this the Gibbs construction is degenerate.

While $J$ is selected by the composition law and calibration, the Shannon entropy functional used in the next section is not an additional axiom. It is recovered by counting microstates: for an array of $N$ independent outcomes drawn from $K$ types, enumerate the integer compositions $\mathbf{n}=(n_1,\ldots,n_K)$ compatible with a prescribed mean of an intensive observable, take the logarithm of the resulting multinomial coefficient, and pass to the large-$N$ limit. The leading term is $N\,H(\hat p)$ with $\hat p_k=n_k/N$ the empirical frequency vector. The constrained multinomial entropy theorem of this section is the finite-state, method-of-types instance of Sanov's theorem, in the sense of Cover and Thomas~\cite{CoverThomas2006}. Density hypotheses on admissible system sizes are required because on a discrete sample space the affine constraint slice may not be reached by integer compositions for every~$N$; Section~\ref{sec:constrained-types} states the lifting and density properties needed.

\subsection{Sample of observations and empirical types}
Consider a sample of $N$ independent observations, each reporting a ratio $r_i\in\{r^{(1)},\ldots,r^{(K)}\}$ drawn from $K$ distinct values. The ordered sequence $(r_1,\ldots,r_N)$ is a sample, and the unordered tally $\mathbf{n}=(n_1,\ldots,n_K)$, with $n_k\ge 0$ counting how many of the $N$ observations report value $r^{(k)}$ and with $\sum_k n_k=N$, is its empirical type. Two samples with the same type differ only by the order of the observations; we count both.

The number of microstates compatible with a given macrostate $\mathbf{n}$ is the multinomial coefficient
\begin{equation}
  W(\mathbf{n})=\frac{N!}{n_1!\,n_2!\cdots n_K!},
  \label{eq:multinomial}
\end{equation}
counting the arrangements of $N$ distinguishable events into $K$ classes of sizes $n_1,\ldots,n_K$. The macrostate supports the empirical distribution $\hat p_k=n_k/N\in[0,1]$ on $\{1,\ldots,K\}$.

\subsection{Multinomial weights and Shannon entropy}
Applying Stirling's approximation $\log(m!)=m\log m-m+O(\log m)$ to each factorial in~\eqref{eq:multinomial} and rewriting in terms of $\hat p$ gives an asymptotic identity in which the implicit constant can be taken uniform over macrostates with $n_k\ge 1$ for all $k$ (boundary cases where some $n_k$ vanishes contribute only lower-order corrections under the convention $0\log 0:=0$):
\begin{equation}
  \frac{1}{N}\log W(\mathbf{n})= -\sum_{k=1}^K \hat p_k\log \hat p_k + O\Bigl(\frac{K\log N}{N}\Bigr)
  =: H_{\mathrm{type}}(\hat p)+O\Bigl(\frac{K\log N}{N}\Bigr).
  \label{eq:stirling}
\end{equation}
$H_{\mathrm{type}}$ denotes Shannon entropy at the empirical frequency vector $\hat p\in\Delta_K$. The same functional on a generic finite outcome space $\Omega$ is denoted $H$ in Sections~\ref{sec:gibbs}--\ref{sec:fe}; the subscript ``type'' is retained only here, where $\hat p$ is an integer-rationally-quantized empirical type. The identity~\eqref{eq:stirling} is the dimensionless Boltzmann principle $S=k_{\mathrm{B}}\log W$~\cite{Boltzmann1877,Ellis1985,CoverThomas2006} on the per-event scale, with $k_{\mathrm{B}}$ absorbed into the choice of nat units; it is the standard finite-state derivation of Shannon entropy as the rate function of a multinomial array, restricted to the affine cost shell selected by the constraint of Section~\ref{sec:constrained-types}.

The asymptotic statement~\eqref{eq:stirling} hides a non-asymptotic Stirling bound that is needed for assessing the regime where $K$ grows with $N$. We record the explicit bound below; the proof reduces to Robbins' two-sided Stirling inequalities~\cite{Robbins1955} applied to each factorial.

\begin{proposition}[Non-asymptotic Stirling bound on $\log W$]
\label{prop:stirling-rate}
For all $N\ge 1$, all $K\ge 2$, and all $\mathbf{n}\in\mathbb{Z}_{\ge 0}^K$ with $\sum_k n_k=N$ and $K^\dagger:=\#\{k:n_k\ge 1\}$ the number of nonempty types, the multinomial weight $W(\mathbf{n})=N!/\prod_k n_k!$ satisfies
\begin{equation}
  \Bigl\lvert\frac{1}{N}\log W(\mathbf{n})-H(\hat p)\Bigr\rvert
  \le \frac{K^\dagger-1}{2N}\log N + \frac{R^\star(K^\dagger)}{N},
  \label{eq:stirling-rate}
\end{equation}
where
\[
  R^\star(K^\dagger):=\frac{K^\dagger-1}{2}\log(2\pi)+\frac{1}{12}+\frac{K^\dagger}{12}.
\]
In particular, when $K^\dagger\le K$ is fixed, the right-hand side of~\eqref{eq:stirling-rate} is $(K-1)\log N/(2N)+O(K/N)$, recovering~\eqref{eq:stirling} with explicit constants. When $K^\dagger$ grows with $N$, the bound degrades: \eqref{eq:stirling-rate} is non-trivial only when $K\log N/N\to 0$, equivalently $N/(K\log N)\to\infty$.
\end{proposition}
\begin{proof}
Let $S=\{k:n_k\ge 1\}$, so $|S|=K^\dagger$. Robbins' inequalities~\cite{Robbins1955} give, for every $m\ge 1$,
\begin{equation}
  \log m! = m\log m-m+\tfrac12\log(2\pi m)+\varepsilon_m,\qquad 0<\varepsilon_m<\frac{1}{12m}.
  \label{eq:robbins}
\end{equation}
Applying~\eqref{eq:robbins} to $N!$ and to $n_k!$ for $k\in S$, and using $\log 0!=0$ for empty types, gives
\[
  \log W(\mathbf{n})
  =N\log N-\sum_{k\in S} n_k\log n_k
  +\tfrac12\Bigl(\log(2\pi N)-\sum_{k\in S}\log(2\pi n_k)\Bigr)
  +\varepsilon_N-\sum_{k\in S}\varepsilon_{n_k}.
\]
Since
\[
  H(\hat p)=\log N-\frac{1}{N}\sum_{k\in S} n_k\log n_k,
\]
we obtain
\[
  \frac{1}{N}\log W(\mathbf{n})-H(\hat p)
  =\frac{1}{2N}\Bigl(\log(2\pi N)-\sum_{k\in S}\log(2\pi n_k)\Bigr)
  +\frac{1}{N}\Bigl(\varepsilon_N-\sum_{k\in S}\varepsilon_{n_k}\Bigr).
\]
The logarithmic term is bounded in absolute value by
\[
  \frac{K^\dagger-1}{2N}\log N+\frac{K^\dagger-1}{2N}\log(2\pi),
\]
because (i) $1\le n_k\le N$ for $k\in S$ implies $0\le\sum_{k\in S}\log n_k\le K^\dagger\log N$, so $|\log N-\sum_{k\in S}\log n_k|\le (K^\dagger-1)\log N$, and (ii) the constant pieces collect into $|{-(K^\dagger-1)\log(2\pi)}|=(K^\dagger-1)\log(2\pi)$. The Robbins remainders satisfy
\[
  \Bigl|\varepsilon_N-\sum_{k\in S}\varepsilon_{n_k}\Bigr|
  \le \frac{1}{12}+\frac{K^\dagger}{12}.
\]
Combining these estimates yields the stated bound.
\end{proof}

For the three-state example of Section~\ref{sec:example} ($K=3$) at $N=100$, taking $K^\dagger=3$ gives the leading term $(K^\dagger-1)\log N/(2N)=\log(100)/100\approx 0.046$ and the Robbins remainder $R^\star(3)/N=(\log(2\pi)+1/12+1/4)/100\approx 0.022$, so the right-hand side of~\eqref{eq:stirling-rate} is at most $\approx 0.068$ nats per observation. This is a uniform absolute error bound on the Stirling approximation to $N^{-1}\log W(\mathbf{n})$, not a relative error bound on the Gibbs weights or on the rate function in any state-dependent sense; thus the rate-function approximation $\frac1N\log W(\mathbf{n})\approx H(\hat p)$ is controlled to within this additive margin at $N=100$. The bound is informative in the fixed-$K$ regime $N\gg K\log N$ and degrades when $K\log N$ approaches $N$ (Remark~\ref{rem:curse}). As a working rule of thumb in the fixed-$K$ regime, $N\gtrsim 30 K\log K$ keeps the right-hand side of~\eqref{eq:stirling-rate} below $\log K/10$.

To leading order in $N$, the number of microstates supporting a given empirical composition $\hat p$ thus grows exponentially at rate $H(\hat p)$. Without constraints, every composition is admissible and the maximizer of $W$ is the maximizer of $H$, namely the uniform distribution. Imposing a linear constraint on the empirical mean of an intensive quantity restricts admissible macrostates to a strict subset, and the maximum of $H$ over this subset selects an empirical distribution that is generally non-uniform. The constrained multinomial entropy theorem below is the finite-state formalization of this picture and the bridge to the Gibbs variational principle of Section~\ref{sec:gibbs}.

\begin{remark}[Regime of validity: fixed $K$, $N\to\infty$]
\label{rem:curse}
The asymptotic results of this section, including Theorem~\ref{thm:types} and Proposition~\ref{prop:stirling-rate}, are stated in the limit $N\to\infty$ at fixed finite outcome space $\Omega$ (equivalently fixed $K=\lvert\Omega\rvert$). This is the inferential regime in which the number of independent observations grows while the underlying system, including its state count, remains fixed; it differs from the classical thermodynamic limit, in which the system size and hence the state count grow, and from large-deviation regimes where $K=K(N)\to\infty$ jointly with $N$. The bound~\eqref{eq:stirling-rate} is non-vacuous precisely when $N\gg K\log N$, equivalently $N\gtrsim K\log K$ for the validity of~\eqref{eq:stirling} as a leading-order asymptotic, and degrades to vacuity once $K\log N$ approaches $N$. As a consequence, for models in which $K$ is exponential in a system size $L$ (e.g., $K=2^L$ for an $L$-spin system) the regime $N\gg K\log N$ requires $N\gtrsim L\,2^L$, far beyond what a realistic sample can support; Theorem~\ref{thm:types} therefore applies to coarse-grained or effective-state-count models with $K$ fixed (or with the active type count $K^\dagger$ controlled), but not as a tight bound for spatially extended lattice models.
\end{remark}

\subsection{Constrained types and admissible system sizes}
\label{sec:constrained-types}
The unconstrained maximizer of $H$ is uniform; the next two subsections specify the constrained problem (this subsection) and prove that constrained multinomial weights are again maximized by an entropy maximizer, now on a restricted slice of the simplex (Section~\ref{sec:constrained-multinomial}). Two ingredients are needed: a definition of the set of integer macrostates compatible with the constraint at fixed system size~$N$, and a verification that this set is nonempty often enough to support a meaningful $N\to\infty$ limit. The first is straightforward; the second is a Diophantine subtlety addressed by Definition~\ref{def:denom-lattice} below.

Heuristically: the continuous constraint set $\Gamma\subset\Delta_K$ is an affine slice, and rational types $p\in\mathbb{Q}^K\cap\Gamma$ are dense in it; but ``$N$ events with type $p$'' requires the integers $Np_k$ to actually be integers. The basic arithmetic observation recorded after~\eqref{eq:macro-feasible-set} below gives the lift, and Definition~\ref{def:denom-lattice} specifies the subsequence of $N$'s along which lifts of arbitrarily entropy-near-maximal rational types are simultaneously available.

The relevant constraint is the standard one of equilibrium statistical mechanics: a fixed average value of an intensive scalar cost. With cost values $g_k:=J(r^{(k)})$, the empirical mean cost at macrostate $\mathbf{n}$ is $\sum_k \hat p_k g_k=N^{-1}\sum_k n_k g_k$, and fixing this mean to a target value $\Ebar$ defines an affine slice of the simplex.

Concretely, fix $g=(g_1,\ldots,g_K)\in\R^K$ and $\Ebar\in\R$. For each $N\in\mathbb{N}$, define the set of nonnegative integer compositions on the affine cost shell by
\begin{equation}
\label{eq:macro-feasible-set}
  \mathcal{F}_N(g,\Ebar):=\Bigl\{\mathbf{n}\in\mathbb{Z}_{\ge 0}^K:\sum_{k=1}^K n_k=N,\ \sum_{k=1}^K n_k g_k=N\Ebar\Bigr\}.
\end{equation}
This is the discrete object that the constrained multinomial theorem of Section~\ref{sec:constrained-multinomial} will count. The pair $(g,\Ebar)$ is macro-feasible if $\mathcal{F}_N(g,\Ebar)\neq\emptyset$ for at least one positive integer~$N$; equivalently, setting $p^\dagger_k=n_k/N$ for any feasible $\mathbf{n}$, there exists a rational type $p^\dagger\in\Delta_K\cap\mathbb{Q}^K$ with $\sum_k p^\dagger_k g_k=\Ebar$. Macro-feasibility is therefore a Diophantine condition on $(g,\Ebar)$, not a geometric one on the constraint slice.

Although the continuous constraint slice $\Gamma=\{p\in\Delta_K:\sum_k p_k g_k=\Ebar\}$ may be a nonempty subset of the simplex, the integer equation $\sum_k n_k g_k=N\Ebar$ need not be solvable for every~$N$, so $\mathcal{F}_N(g,\Ebar)$ can be empty. A naive limit $N\to\infty$ may wander through values of $N$ at which there is nothing to count, so the exact-shell asymptotic statement in Corollary~\ref{cor:types-exact} below is phrased along subsequences for which integer compositions of the right type both exist and accumulate near the entropy maximizer. The relevant lifting and density properties are stated in Definition~\ref{def:denom-lattice} below; this is the only Diophantine input to the corollary. As a basic arithmetic observation: if $q\in\Delta_K\cap\mathbb{Q}^K$ has common denominator $Q\in\mathbb{N}$ (so $Qq_k\in\mathbb{Z}_{\ge 0}$ for every $k$) and $\sum_k q_k g_k=\Ebar$ with $g\in\mathbb{Q}^K,\Ebar\in\mathbb{Q}$, then for any $N$ divisible by $Q$ the composition $n_k:=Nq_k$ is an element of $\mathcal{F}_N(g,\Ebar)$, because $\sum_k n_k=N\sum_k q_k=N$ and $\sum_k n_k g_k=N\sum_k q_k g_k=N\Ebar$.

\begin{definition}[Denominator-exhaustive sequence]
\label{def:denom-lattice}
With $h_\star:=\max_{p\in\Gamma}H(p)$ and $\mathcal{I}:=\{N\in\mathbb{N}:\mathcal{F}_N\neq\emptyset\}$, a sequence $(N_j)_{j\ge 1}\subset\mathcal{I}$ is denominator-exhaustive if there exists a countable family $\{q^{(m)}\}_{m\ge 1}\subset\Gamma\cap\mathbb{Q}^K$ with $H(q^{(m)})\to h_\star$ as $m\to\infty$ and with common denominators $Q^{(m)}$ (so $Q^{(m)} q^{(m)}_k\in\mathbb{Z}_{\ge 0}$ for every $k$) such that, for every fixed $m$, $Q^{(m)}$ divides $N_j$ for all sufficiently large $j$ (i.e.\ there exists $J(m)\in\mathbb{N}$ with $Q^{(m)}\mid N_j$ for every $j\ge J(m)$).
\end{definition}

Eventual divisibility is the strengthening relative to the weaker ``divides infinitely many $N_j$'' formulation: it ensures that the lifted compositions $\mathbf{n}^{(j,m)}=N_j q^{(m)}\in\mathcal{F}_{N_j}$ are available for every $j\ge J(m)$, so the lower bound below can be obtained along the full tail of the sequence rather than only along a sparser sub-subsequence. Concrete admissible patterns include $N_j=j!$ (eventually divisible by every fixed denominator) and arithmetic progressions $N_j=jQ_0$ chosen so that $Q_0$ is a common multiple of the denominators in an approximating family $\{q^{(m)}\}$. For rational $(g,\Ebar)$, periodicity of $\mathcal{I}$ can also be read from the Smith normal form of the $2\times K$ constraint matrix~\cite{Schrijver1998}; this observation is not needed elsewhere.

\subsection{Constrained multinomial theorem}
\label{sec:constrained-multinomial}
Among macrostates compatible with a prescribed average of the costs $g_k=J(r^{(k)})$, the multinomial weight is asymptotically maximized by empirical compositions that maximize Shannon entropy on the corresponding affine slice. The principal asymptotic statement is Theorem~\ref{thm:types} below, which allows arbitrary real $(g,\Ebar)$ by relaxing the hard equality $\sum_k n_k g_k=N\Ebar$ to a shrinking tolerance window (the \emph{soft shell}). The exact integer equality case for rational data is recovered as Corollary~\ref{cor:types-exact}. The proofs use Stirling's formula uniformly in types together with continuity of $H$ on the simplex~\cite{CoverThomas2006,CsiszarKorner1981,Sanov1957,DemboZeitouni2010,Touchette2009}, in the form needed for Section~\ref{sec:gibbs}.

For real-valued $(g,\Ebar)$, define the soft-shell feasible set
\begin{equation}
  \mathcal{F}_N^{\delta_N}(g,\Ebar):=\Bigl\{\mathbf{n}\in\mathbb{Z}_{\ge 0}^K:\sum_k n_k=N,\ \bigl\lvert N^{-1}\sum_k n_k g_k-\Ebar\bigr\rvert\le\delta_N\Bigr\},
  \label{eq:soft-shell}
\end{equation}
with tolerance $\delta_N>0$.

\begin{theorem}[Constrained multinomial weights; soft shell]
\label{thm:types}
Fix $K\ge 2$, $g=(g_1,\ldots,g_K)\in\R^K$, and $\Ebar\in\R$. Assume the affine slice
\[
  \Gamma:=\{p\in\Delta_K:\sum_{k=1}^K p_k g_k=\Ebar\}
\]
is nonempty, and set $h_\star:=\max_{p\in\Gamma}H(p)$. Let $\delta_N>0$ satisfy $\delta_N\to 0$ and $N\delta_N\to\infty$---the minimal tolerance scaling under which the argument below yields the limit~\eqref{eq:soft-shell-limit} along the full sequence $N\to\infty$ (rather than merely along favorable subsequences). Then $\mathcal{F}_N^{\delta_N}(g,\Ebar)$ is nonempty for all sufficiently large $N$, and
\begin{equation}
  \lim_{N\to\infty}\frac{1}{N}\log\max_{\mathbf{n}\in\mathcal{F}_N^{\delta_N}(g,\Ebar)}W(\mathbf{n})=h_\star.
  \label{eq:soft-shell-limit}
\end{equation}
\end{theorem}
\begin{proof}
For $\delta\ge 0$, write
\[
  \Gamma_\delta:=\{p\in\Delta_K:|\langle g\rangle_p-\Ebar|\le\delta\},
  \qquad \langle g\rangle_p:=\sum_k p_k g_k.
\]
The sets $\Gamma_\delta$ are compact and decrease to $\Gamma$ as $\delta\downarrow 0$. Since $H$ is continuous on the compact simplex, 
\[
  \sup_{p\in\Gamma_\delta}H(p)\longrightarrow h_\star
  \qquad(\delta\downarrow 0).
\]
Indeed, if this failed, there would be $\eta>0$, $\delta_m\downarrow 0$, and $p^{(m)}\in\Gamma_{\delta_m}$ with $H(p^{(m)})\ge h_\star+\eta$; a convergent subsequence would have a limit $p^\infty\in\Gamma$ by compactness and continuity of $\langle g\rangle_p$, and continuity of $H$ would give $H(p^\infty)\ge h_\star+\eta$, contradicting the definition of $h_\star$.

For the upper bound, any $\mathbf{n}\in\mathcal{F}_N^{\delta_N}(g,\Ebar)$ has empirical type $\hat p=\mathbf{n}/N\in\Gamma_{\delta_N}$. Proposition~\ref{prop:stirling-rate} gives, uniformly in $\mathbf{n}$ at fixed $K$,
\[
  \frac1N\log W(\mathbf{n})\le H(\hat p)+o(1)
  \le \sup_{p\in\Gamma_{\delta_N}}H(p)+o(1)=h_\star+o(1),
\]
so the limsup in~\eqref{eq:soft-shell-limit} is at most $h_\star$.

For the lower bound, choose $p^\star\in\Gamma$ with $H(p^\star)=h_\star$; such a maximizer exists by compactness of $\Gamma$. For each $N$, choose an integer vector $\mathbf{n}^{(N)}$ with $\sum_k n^{(N)}_k=N$ and
\[
  \left\|\frac{\mathbf{n}^{(N)}}{N}-p^\star\right\|_1\le \frac{K}{N},
\]
obtained, for example, by taking floors of $Np^\star_k$ and distributing the remaining mass among the largest fractional parts. Then
\[
  \left|\left\langle g\right\rangle_{\mathbf{n}^{(N)}/N}-\Ebar\right|
  =\left|\left\langle g\right\rangle_{\mathbf{n}^{(N)}/N}-\left\langle g\right\rangle_{p^\star}\right|
  \le \|g\|_\infty\,\left\|\frac{\mathbf{n}^{(N)}}{N}-p^\star\right\|_1
  \le \frac{K\|g\|_\infty}{N}.
\]
Since $N\delta_N\to\infty$, the last quantity is at most $\delta_N$ for all sufficiently large $N$. Thus $\mathbf{n}^{(N)}\in\mathcal{F}_N^{\delta_N}(g,\Ebar)$ eventually. By continuity of $H$ and Proposition~\ref{prop:stirling-rate},
\[
  \frac1N\log W(\mathbf{n}^{(N)})
  =H(\mathbf{n}^{(N)}/N)+o(1)\longrightarrow H(p^\star)=h_\star.
\]
This gives the matching liminf and proves~\eqref{eq:soft-shell-limit}.
\end{proof}

\begin{corollary}[Exact affine shell under rational costs]
\label{cor:types-exact}
Fix $K\ge 2$ and ratio values $r^{(1)},\ldots,r^{(K)}$ with $g_k:=J(r^{(k)})$ not all equal. Write $\mathcal{F}_N:=\mathcal{F}_N(g,\Ebar)$ as in~\eqref{eq:macro-feasible-set}. Assume $(g,\Ebar)$ is macro-feasible, that $\min_k g_k<\Ebar<\max_k g_k$, and that $g_1,\ldots,g_K,\Ebar\in\mathbb{Q}$, so that $\Gamma:=\{p\in\Delta_K:\sum_k p_k g_k=\Ebar\}$ is a rational affine slice and $\Gamma\cap\mathbb{Q}^K$ is dense in $\Gamma$. Set
\[
  h_\star:=\max_{p\in\Gamma} H(p),
\]
and let $\mathcal{I}:=\{N\in\mathbb{N}:\mathcal{F}_N\neq\emptyset\}$. Let $(N_j)_{j\ge 1}\subset\mathcal{I}$ be a denominator-exhaustive sequence (Definition~\ref{def:denom-lattice}). For each $N_j$, let $\mathbf{n}^\star\in\mathcal{F}_{N_j}$ maximize $W(\mathbf{n})$ on $\mathcal{F}_{N_j}$, and set $\hat p^{(N_j)}_k=n^\star_k/N_j$. Then every subsequential limit $\hat p^\infty$ of $(\hat p^{(N_j)})$ in $\Delta_K$ satisfies $H(\hat p^\infty)=h_\star$, and
\[
  \lim_{j\to\infty}\frac{1}{N_j}\log\max_{\mathbf{n}\in\mathcal{F}_{N_j}}W(\mathbf{n})=h_\star.
\]
\end{corollary}
\begin{proof}
Fix any sequence $\delta_N\downarrow 0$ with $N\delta_N\to\infty$. For every $N$ and every $\mathbf{n}\in\mathcal{F}_N$, the exact equality $\sum_k n_k g_k=N\Ebar$ implies $\mathbf{n}\in\mathcal{F}_N^{\delta_N}(g,\Ebar)$, hence
\[
  \max_{\mathbf{n}\in\mathcal{F}_N}W(\mathbf{n})\le \max_{\mathbf{n}\in\mathcal{F}_N^{\delta_N}(g,\Ebar)}W(\mathbf{n}).
\]
Applying Theorem~\ref{thm:types} along the subsequence $(N_j)$ gives
\[
  \limsup_{j\to\infty}\frac{1}{N_j}\log\max_{\mathbf{n}\in\mathcal{F}_{N_j}}W(\mathbf{n})\le h_\star.
\]

For the lower bound, write $H_{N_j}:=\max_{\mathbf{n}\in\mathcal{F}_{N_j}}\frac{1}{N_j}\log W(\mathbf{n})$ and fix $\varepsilon>0$. Under the rationality hypothesis, $\Gamma$ is a rational polytope with nonempty relative interior whenever $\min_k g_k<\Ebar<\max_k g_k$, so $\Gamma\cap\mathbb{Q}^K$ is dense in $\Gamma$ and continuity of $H$ on $\Delta_K$ guarantees the existence of a rational near-maximizer $q^{(\varepsilon)}\in\Gamma\cap\mathbb{Q}^K$ with $H(q^{(\varepsilon)})\ge h_\star-\varepsilon$. By Definition~\ref{def:denom-lattice}, choosing $m$ large enough that $H(q^{(m)})\ge h_\star-\varepsilon$, the denominator $Q^{(m)}$ divides $N_j$ for all $j\ge J(m)$, so for every such $j$ the lifted composition $\mathbf{n}^{(j)}:=N_j q^{(m)}\in\mathcal{F}_{N_j}$ has $\tilde p^{(j)}=\mathbf{n}^{(j)}/N_j\equiv q^{(m)}$ exactly. Continuity of $H$ gives $H(\tilde p^{(j)})=H(q^{(m)})$, and Stirling's formula yields, for all $j\ge J(m)$,
\[
  \frac{1}{N_j}\log W(\mathbf{n}^{(j)}) = H(q^{(m)}) + o(1) \ge h_\star - \varepsilon + o(1).
\]
Since $\mathbf{n}^\star$ maximizes $W$ on $\mathcal{F}_{N_j}$, $H_{N_j}\ge h_\star-\varepsilon+o(1)$ for all $j\ge J(m)$. Because eventual divisibility holds at every $m$, letting $\varepsilon_m\downarrow 0$ along the family $\{q^{(m)}\}$ gives $\liminf_{j\to\infty} H_{N_j}\ge h_\star$ along the full sequence $(N_j)$. Combined with the upper bound, $\lim_j H_{N_j}=h_\star$. Any subsequential limit of maximizers attains $h_\star$ because $H$ is continuous on $\Delta_K$ (with $0\log 0:=0$).
\end{proof}

\begin{remark}
\label{rem:types-general}
Theorem~\ref{thm:types} is the full-sequence soft-shell method-of-types statement for arbitrary real $(g,\Ebar)$ with nonempty $\Gamma$. Corollary~\ref{cor:types-exact} is the exact-shell specialization: it holds along any sequence $(N_j)\subset\mathcal{I}$ for which there exist $\mathbf{n}^{(j)}\in\mathcal{F}_{N_j}$ with $H(\mathbf{n}^{(j)}/N_j)\to h_\star$, and denominator-exhaustive sequences (Definition~\ref{def:denom-lattice}) supply the Diophantine lower bound in the rational case. Without such an attainability hypothesis, sparse subsequences exist for which feasible types fail to approach the entropy maximizer, and convergence along the full sequence $N\to\infty$ with $\mathcal{F}_N\neq\emptyset$ cannot be claimed for the exact shell.
\end{remark}

Theorem~\ref{thm:types} is the finite-$K$ method-of-types statement that maximizing the multinomial weight at approximately fixed average cost is equivalent to maximizing Shannon entropy on the constraint slice~\cite{CoverThomas2006,CsiszarKorner1981,Sanov1957,DemboZeitouni2010}. Corollary~\ref{cor:types-exact} packages the classical exact rational slice; the eventual-divisibility bookkeeping of Definition~\ref{def:denom-lattice} is needed only there. Quantitative finite-$N$, finite-$K$ control is given by Proposition~\ref{prop:stirling-rate}.

In summary, the cost functional $J$ is fixed by the RCL once the ratios $r^{(k)}$ are specified, the entropy functional $H$ is supplied by combinatorial counting under Stirling's formula, and the constrained variational problem of maximizing $H(p)$ over $\Gamma=\{p\in\Delta_K:\langle g\rangle_p=\Ebar\}$ is the asymptotic shadow of maximizing $W(\mathbf{n})$ over $\mathcal{F}_N(g,\Ebar)$. Section~\ref{sec:gibbs} converts this equivalence into an explicit Gibbs measure via Lagrange optimization.

\subsection{Modeling inputs: ratios, costs, and physical interpretation}
\label{sec:ledger-bridge}
The finite-state framework treats the microscopic labels $r_\omega$ as inputs: they are supplied by a ratio-cost model rather than derived from a Hamiltonian or any other underlying dynamical postulate. Each outcome $\omega$ is assigned a positive dimensionless ratio $r_\omega$ relative to a chosen reference scale, and the RCL then fixes the scalar cost as $X_\omega=J(r_\omega)$.

Once the ratios are specified, the energy-like cost in the canonical variational problem is determined; mapping the ratios $r_\omega$ to a material Hamiltonian, a transition rate, an affinity, or an experimentally measurable imbalance is model-specific and lies outside the present scope. In a reversible two-state chain, detailed balance gives the input ratio $r_b=W_{ab}/W_{ba}=\pi_b/\pi_a$~\cite{LevinPeresWilmer2017}, so the corresponding cost vector is $(0,J(r_b))$; the three-state example of Section~\ref{sec:example} illustrates the construction in a nontrivial setting, including the Gibbs distribution at a prescribed mean cost and its comparison with the chain's stationary law.

\section{Variational derivation of the finite-state Gibbs law}
\label{sec:gibbs}
The previous sections supply two ingredients: the closed-form cost $J(x)=\tfrac12(x+x^{-1})-1$ from Theorem~\ref{thm:wz-uniqueness}, and Shannon entropy as the asymptotic rate function for multinomial weights on the affine cost shell (Theorem~\ref{thm:types}). The convex-analytic step assembles $H$ and $X_\omega=J(r_\omega)$ into a Lagrangian, identifies its unique critical point on the cost shell, and identifies that critical point as the canonical Gibbs measure. The argument is the standard convex duality of the canonical ensemble, applied here with $X_\omega$ selected upstream from the input ratios once the RCL assumptions are adopted. The section makes no use of the explicit form of $J$ from~\eqref{eq:J-closed}: it works for any non-constant cost vector $X_\omega$ on a finite outcome space, and the role of the RCL is only to fix that vector. A reader interested only in the canonical ensemble on a finite space may skip the dependence on $J$ and read $X$ as a generic real-valued observable.

\subsection{Setup: entropy, mean cost, and inverse temperature}
\textbf{Standing assumptions for Sections~\ref{sec:gibbs}--\ref{sec:fe}:} $\Omega$ is a finite set with $\lvert\Omega\rvert\ge 2$; positive ratios $r_\omega>0$ are specified for each $\omega\in\Omega$; the cost is $X_\omega:=J(r_\omega)$; and $X$ is non-constant on $\Omega$ (i.e., not all $X_\omega$ are equal). The constraint value $\Ebar$ satisfies $\min_\omega X_\omega<\Ebar<\max_\omega X_\omega$. The thermodynamic convention $\TR>0$ (equivalently $\beta>0$) is in force, so $\Ebar$ lies below the uniform-distribution mean $\langle X\rangle_{p^{(0)}}$ (Proposition~\ref{prop:inv-beta}).

Two functionals on this set play a central role. The first is Shannon entropy,
\[
  H(q) := -\sum_{\omega\in\Omega} q_\omega\log q_\omega,
\]
written with the same symbol $H$ as in Section~\ref{sec:counting} (there applied to empirical types $\hat p\in\Delta_K$). The second is the expected cost,
\[
  \langle X\rangle_q := \sum_{\omega\in\Omega} q_\omega X_\omega,
\]
which plays the role of $\sum_k \hat p_k g_k$ from Section~\ref{sec:counting}. The variational problem is to maximize $H(q)$ subject to a fixed value of $\langle X\rangle_q=\Ebar$ together with the normalization $\sum_\omega q_\omega=1$.

The inverse temperature $\beta$ is the Lagrange multiplier conjugate to the mean-cost constraint. In the canonical (positive-temperature) regime $\beta>0$, set the recognition temperature $\TR:=1/\beta$ (dimensionless unless an additional physical scale is supplied). The candidate maximizer then takes the familiar Gibbs form
\begin{equation}
  p_\omega=\frac{e^{-X_\omega/\TR}}{Z},\qquad
  Z=\sum_{\omega'\in\Omega} e^{-X_{\omega'}/\TR},
  \label{eq:gibbs}
\end{equation}
where $Z$ is the partition function. The next two subsections verify that~\eqref{eq:gibbs} is the unique maximizer and that $\TR$ is uniquely determined by the constraint value~$\Ebar$.

\subsection{Existence and uniqueness of the dual parameter}
Working with $\beta$ rather than $\TR$ allows a uniform treatment of the positive- and negative-temperature regimes. For each $\beta\in\R$, define the partition function
\[
  Z(\beta):=\sum_{\omega\in\Omega} e^{-\beta X_\omega}
\]
and the parametric Gibbs law $p^{(\beta)}_\omega:=e^{-\beta X_\omega}/Z(\beta)$. The family $\beta\mapsto p^{(\beta)}$ is the exponential family canonically associated with the sufficient statistic~$X$, and the cumulant generating function $\log Z(\beta)$ encodes its differential structure.

\begin{lemma}[Strict convexity of $\log Z$]
\label{lem:logZ}
Under the standing assumption that $X$ is non-constant on $\Omega$, $\log Z$ is strictly convex on $\R$, and
\[
  \frac{\mathrm{d}}{\mathrm{d}\beta}\log Z(\beta)=-\langle X\rangle_{p^{(\beta)}},\qquad
  \frac{\mathrm{d}^2}{\mathrm{d}\beta^2}\log Z(\beta)=\mathrm{Var}_{p^{(\beta)}}(X)>0.
\]
Consequently $\beta\mapsto\langle X\rangle_{p^{(\beta)}}$ is strictly decreasing, with
\[
  \lim_{\beta\to-\infty}\langle X\rangle_{p^{(\beta)}}=\max_{\omega\in\Omega}X_\omega,
  \qquad
  \lim_{\beta\to+\infty}\langle X\rangle_{p^{(\beta)}}=\min_{\omega\in\Omega}X_\omega.
\]
\end{lemma}
\begin{proof}
Differentiate $\log Z$ in $\beta$; the variance formula is the standard exponential-family identity~\cite{Ellis1985}. Strict positivity of the variance for non-constant $X$ implies strict convexity of $\log Z$ and hence a strictly decreasing derivative $-\langle X\rangle_{p^{(\beta)}}$. The limiting values follow from the concentration of $p^{(\beta)}$ on the argmax (respectively argmin) of $X$ as $\beta\to-\infty$ (respectively $\beta\to+\infty$).
\end{proof}

As in the cumulant-generating-function picture for exponential families~\cite{Ellis1985}, $\log Z$ is a strictly convex potential whose Legendre transform encodes the entropy--mean-cost tradeoff. Strict convexity makes the conjugate variable $\langle X\rangle_{p^{(\beta)}}$ a strictly monotone function of $\beta$, taking each value in the open interval $(\min X,\max X)$ exactly once. Each admissible constraint value $\Ebar$ therefore corresponds to a unique Lagrange multiplier~$\beta$, as stated in the proposition below.

\begin{proposition}[Existence and uniqueness of the dual parameter $\beta$]
\label{prop:inv-beta}
Assume $X$ is not constant on $\Omega$. For every $\Ebar$ strictly between $\min_\omega X_\omega$ and $\max_\omega X_\omega$, there exists a unique $\beta\in\R$ such that $\langle X\rangle_{p^{(\beta)}}=\Ebar$. Moreover, $\beta>0$ if and only if $\Ebar<\langle X\rangle_{p^{(0)}}$ (the mean cost under the uniform distribution), and $\beta<0$ if and only if the reverse strict inequality holds. In the thermodynamic convention $\TR>0$ used in Sections~\ref{sec:gibbs}--\ref{sec:fe}, we restrict attention to $\beta>0$, equivalently $\Ebar$ below the uniform average.
\end{proposition}
\begin{proof}
By Lemma~\ref{lem:logZ}, $\beta\mapsto\langle X\rangle_{p^{(\beta)}}$ is continuous and strictly decreasing with the stated limits, hence bijects $\R$ onto $(\min X,\max X)$. The sign of $\beta$ is determined by comparing $\Ebar$ to $\langle X\rangle_{p^{(0)}}$, which is the uniform average.
\end{proof}

\subsection{Entropy maximization and Gibbs weights}
With $\beta$ uniquely fixed by Proposition~\ref{prop:inv-beta}, the variational principle can be stated and proved. The argument is the classical Lagrange-multiplier calculation for maximum entropy at fixed mean cost: the candidate maximizer produced by stationarity coincides with the parametric Gibbs law $p^{(\beta)}$ from Lemma~\ref{lem:logZ}, and strict concavity of Shannon entropy promotes this critical point to the unique global maximizer.

\begin{proposition}[Finite-state Gibbs variational principle; classical finite-state form]
\label{thm:maxent}
Assume $X$ is not constant on $\Omega$. Let $\Ebar$ lie strictly between $\min_\omega X_\omega$ and $\max_\omega X_\omega$, and let $\beta\in\R$ be the unique dual parameter from Proposition~\ref{prop:inv-beta} such that $\langle X\rangle_{p^{(\beta)}}=\Ebar$. Among all probability distributions on $\Omega$ with $\langle X\rangle_q=\Ebar$ and strictly positive weights, the Gibbs law $p^{(\beta)}_\omega=e^{-\beta X_\omega}/Z(\beta)$ uniquely maximizes the Shannon entropy $H(q)$. In the thermodynamic convention $\TR>0$ used in~\eqref{eq:gibbs} and Section~\ref{sec:fe}, we restrict attention to $\beta>0$, equivalently $\Ebar<\langle X\rangle_{p^{(0)}}$; in this regime $p^{(\beta)}$ coincides with~\eqref{eq:gibbs} for $\TR=1/\beta$.
\end{proposition}
\begin{proof}
Form the Lagrangian
\begin{equation}
  \mathcal{L}(q)=H(q)
  -\beta\Bigl(\sum_{\omega} q_\omega X_\omega-\Ebar\Bigr)
  -\lambda\Bigl(\sum_{\omega} q_\omega-1\Bigr),
\end{equation}
where $\lambda$ is the multiplier conjugate to the normalization constraint. Stationarity in each $q_\omega$ yields the equations $-\log q_\omega-1-\beta X_\omega-\lambda=0$, i.e.\ $\log q_\omega=-\beta X_\omega-(1+\lambda)$ for every $\omega\in\Omega$. Setting $c:=1+\lambda$ (a single $\omega$-independent constant absorbing the multiplier and the constant $-1$ from $\partial H/\partial q_\omega$) and exponentiating gives
\[
  q_\omega=e^{-\beta X_\omega-c},\qquad c\in\R;
\]
imposing $\sum_\omega q_\omega=1$ then forces $e^{-c}=1/Z(\beta)$ and hence $q_\omega=e^{-\beta X_\omega}/Z(\beta)=p^{(\beta)}_\omega$. Strict concavity of $H$ on the interior of the simplex then implies that this interior critical point is the unique maximizer among strictly positive distributions $q$ with $\langle X\rangle_q=\Ebar$~\cite{CoverThomas2006}; the extension to the full closed simplex (where some $q_\omega$ may vanish) when $X$ is non-constant and $\Ebar\in(\min X,\max X)$ follows from the strict concavity of $H$ on $\Delta_\Omega$ together with the fact that $H$ is continuous on the closed simplex (under the convention $0\log 0:=0$): an interior critical point of a strictly concave continuous function on a compact convex set is the unique global maximum on that set whenever the constraint slice has nonempty interior~\cite[\S II.4]{Ellis1985}. Uniqueness of $\beta$ itself is Proposition~\ref{prop:inv-beta}.
\end{proof}

Proposition~\ref{thm:maxent} is the classical Gibbs variational principle on a finite space, transcribed into the cost notation used here. The Gibbs weights are therefore not introduced as an additional assumption: they arise as the unique interior critical point of a strictly concave functional on a linear constraint slice. Lemma~\ref{lem:logZ} is the cumulant-generating-function viewpoint on the canonical ensemble~\cite{Ellis1985}, with the identification $\TR=1/\beta$ used throughout in the positive-temperature regime.

The model-specific input is the cost vector $X_\omega=J(r_\omega)$ fixed by axioms (A1)--(A3) from the ratios $\{r_\omega\}$. Given~$X$, the Gibbs law at prescribed mean cost and the free-energy--Kullback--Leibler identities of Section~\ref{sec:fe} are the standard finite-state objects determined by that cost; Section~\ref{sec:discussion} explains how this assignment fits classical maximum-entropy and large-deviation templates.

\section{Free energy, KL identity, and relaxation}
\label{sec:fe}
The Gibbs law of Section~\ref{sec:gibbs} also admits a second, equally fundamental characterization: it is the unique minimizer of a thermodynamic free-energy functional, and the gap between the free energy at an arbitrary distribution and at the Gibbs reference is proportional to a Kullback--Leibler divergence. This identity is purely algebraic---a direct expansion of definitions---and holds for any probability distribution $q$ on $\Omega$, independently of any variational principle. The connection to relaxation dynamics enters through Proposition~\ref{prop:KL-markov}, which provides a master equation for which $\DKL(\cdot\Vert p)$ decreases monotonically.

The free-energy--Kullback--Leibler identity (Proposition~\ref{thm:fe-kl}) and the Gibbs minimization property (Corollary~\ref{cor:fe-min}) are standard finite-state statements of equilibrium statistical mechanics~\cite{Ellis1985,CoverThomas2006,Touchette2009}; the master-equation monotonicity of $\DKL$ under detailed balance (Proposition~\ref{prop:KL-markov}) is also classical~\cite{Schnakenberg1976,CoverThomas2006,Voigt1981}. The proofs given below fix conventions in the cost notation $X_\omega=J(r_\omega)$ used in Sections~\ref{sec:example} and~\ref{sec:discussion}.

\subsection{Free energy and relative entropy}
The recognition free energy at dimensionless temperature $\TR>0$, in direct analogy with the Helmholtz functional $F=E-TS$ of equilibrium thermodynamics~\cite{Ellis1985}, is
\[
  F_{\mathrm{R}}(q)=\langle X\rangle_q-\TR H(q)
  =\sum_{\omega} q_\omega X_\omega+\TR\sum_{\omega} q_\omega\log q_\omega,
\]
the standard combination of mean cost and temperature-weighted entropy, with the cost $X_\omega=J(r_\omega)$ supplied by the RCL. The directed Kullback--Leibler divergence~\cite{KullbackLeibler1951} between probability distributions $q$ and $p$ with $p_\omega>0$ for all $\omega$ is
\[
  \DKL(q\Vert p)=\sum_{\omega} q_\omega\log\frac{q_\omega}{p_\omega},
\]
under the convention $0\cdot\log 0=0$. The orientation $\DKL(q\Vert p)$, with $q$ in the first argument, is the orientation that arises when $F_{\mathrm{R}}(q)$ is expanded around a Gibbs reference $p$, and is kept fixed below.

\subsection{Free-energy--KL identity}
The identity that ties $F_{\mathrm{R}}$ and $\DKL$ together is purely algebraic. It depends only on the log-affinity of the Gibbs measure in the cost $X$ and requires no appeal to the variational principle of Section~\ref{sec:gibbs}.

\begin{proposition}[Free Energy--KL Identity; classical, see e.g.~\cite{Ellis1985,CoverThomas2006,Touchette2009}]
\label{thm:fe-kl}
Fix $\TR>0$, let $p$ denote the Gibbs measure~\eqref{eq:gibbs} with this $\TR$ and partition function $Z$, and let $F_{\mathrm{R}}$ be defined with the same~$\TR$. Then for any probability distribution $q$ on~$\Omega$,
\begin{equation}
  F_{\mathrm{R}}(q)-F_{\mathrm{R}}(p)=\TR\,\DKL(q\Vert p),
\end{equation}
where $F_{\mathrm{R}}(p)=-\TR\log Z$.
\end{proposition}
\begin{proof}
By definition, $F_{\mathrm{R}}(q)=\sum_\omega q_\omega X_\omega+\TR\sum_\omega q_\omega\log q_\omega$. Since $p_\omega=e^{-X_\omega/\TR}/Z$ by~\eqref{eq:gibbs}, taking logarithms gives $\log p_\omega=-X_\omega/\TR-\log Z$, so $-\TR\log p_\omega=X_\omega+\TR\log Z$. Substituting this identity into the definition of $F_{\mathrm{R}}(p)$ gives
\[
  F_{\mathrm{R}}(p)=\sum_\omega p_\omega X_\omega+\TR\sum_\omega p_\omega\log p_\omega
  =\sum_\omega p_\omega X_\omega-\sum_\omega p_\omega X_\omega-\TR\log Z=-\TR\log Z.
\]
Substituting,
\begin{align*}
  \TR\,\DKL(q\Vert p)
  &=\TR\sum_{\omega} q_\omega(\log q_\omega-\log p_\omega)\\
  &=\TR\sum_{\omega} q_\omega\log q_\omega
  +\sum_{\omega} q_\omega X_\omega+\TR\log Z\\
  &=F_{\mathrm{R}}(q)+\TR\log Z
  =F_{\mathrm{R}}(q)-F_{\mathrm{R}}(p),
\end{align*}
using the identity $F_{\mathrm{R}}(p)=-\TR\log Z$ in the last step.
\end{proof}

The proof is a direct expansion: the cost-linear terms reorganize because $\log p$ is affine in $X$ for the Gibbs reference, and the remainder assembles into $\DKL$. Proposition~\ref{thm:fe-kl} is therefore a structural identity for the exponential family on~$\Omega$, requiring no convex optimization. Two immediate consequences are recorded below: free-energy minimization and, as a remark, the equivalent maximum-entropy statement at fixed cost.

\subsection{Variational consequences}
The first corollary is the free-energy minimization property: the Gibbs measure $p$ minimizes $F_{\mathrm{R}}$, with the gap to any competitor given exactly by $\TR\,\DKL$. This is the dual of Proposition~\ref{thm:maxent} and a one-line consequence of Proposition~\ref{thm:fe-kl}.

\begin{corollary}[Gibbs minimizes free energy]
\label{cor:fe-min}
$F_{\mathrm{R}}(p)\le F_{\mathrm{R}}(q)$ for all probability distributions~$q$, with equality if and only if $q=p$.
\end{corollary}
\begin{proof}
By Proposition~\ref{thm:fe-kl}, $F_{\mathrm{R}}(q)-F_{\mathrm{R}}(p)=\TR\,\DKL(q\Vert p)\ge 0$, since $\TR>0$ and $\DKL\ge 0$. The KL divergence vanishes if and only if $q=p$.
\end{proof}

The maximum-entropy statement of Section~\ref{sec:gibbs} can also be recovered directly from Corollary~\ref{cor:fe-min}, without invoking Lagrange stationarity: at fixed mean cost the only term in $F_{\mathrm{R}}$ that varies is the entropy term, and the sign of that variation is fixed by the free-energy minimization property.

\begin{remark}[Same maximum-entropy statement via the free-energy--KL identity]
\label{cor:entropy-max}
Among all $q$ with $\langle X\rangle_q=\langle X\rangle_p$,
$F_{\mathrm{R}}(q)-F_{\mathrm{R}}(p)=(\langle X\rangle_q-\TR H(q))-(\langle X\rangle_p-\TR H(p))=-\TR(H(q)-H(p))$, and Corollary~\ref{cor:fe-min} gives $F_{\mathrm{R}}(q)\ge F_{\mathrm{R}}(p)$, hence $H(q)\le H(p)$ with equality iff $q=p$. This is Proposition~\ref{thm:maxent} re-derived through the free-energy--KL identity rather than by direct Lagrange stationarity, and is logically identical to Proposition~\ref{thm:maxent}; it is recorded here only as a consistency check on the algebra of Section~\ref{sec:fe}.
\end{remark}

\subsection{Detailed-balance relaxation and monotonicity}
The variational principle is static. Section~\ref{sec:gibbs} shows that the Gibbs measure is the entropy maximum at fixed mean cost; Proposition~\ref{thm:fe-kl} above shows that it is also the unique global minimum of $F_{\mathrm{R}}$ on the simplex; neither statement involves a dynamics. To couple the static identity to a process, we now consider a continuous-time Markov chain on $\Omega$ whose stationary law is $p$ and whose transition rates satisfy detailed balance with respect to~$p$. Along such a flow, the Kullback--Leibler divergence to the stationary distribution is a Lyapunov function, and Proposition~\ref{thm:fe-kl} translates that monotonicity into the corresponding statement for $F_{\mathrm{R}}$.

\begin{proposition}[Monotonicity of $\DKL$ under detailed balance; classical, cf.~\cite{Schnakenberg1976,Voigt1981,CoverThomas2006}]
\label{prop:KL-markov}
Let $Q=(Q_{ij})_{i,j\in\Omega}$ satisfy $Q_{ij}\ge 0$ for $i\neq j$, row sums $\sum_j Q_{ij}=0$, and detailed balance $p_i Q_{ij}=p_j Q_{ji}$ with respect to a fixed Gibbs measure $p>0$ on~$\Omega$. Suppose $t\mapsto q(t)$ is a differentiable curve in the interior of the simplex $\Delta(\Omega)$ that satisfies the master equation
\begin{equation}
  \dot{q}_j(t)=\sum_{i\in\Omega} q_i(t)\,Q_{ij},\qquad j\in\Omega.
  \label{eq:master}
\end{equation}
Define $x_i(t):=q_i(t)/p_i$. Then
\begin{equation}
  \frac{\mathrm d}{\mathrm dt}\,\DKL\bigl(q(t)\Vert p\bigr)
  =-\frac{1}{2}\sum_{i,j\in\Omega} p_i Q_{ij}\,\bigl(x_j(t)-x_i(t)\bigr)\bigl(\log x_j(t)-\log x_i(t)\bigr)\le 0,
  \label{eq:KL-derivative}
\end{equation}
so $t\mapsto \DKL(q(t)\Vert p)$ is nonincreasing. If, in addition, $Q$ is irreducible on $\Omega$ (equivalently, the off-diagonal graph $\{(i,j):i\neq j,\ Q_{ij}>0\}$ is connected as an undirected graph), then by the Perron--Frobenius theorem applied to the rate matrix~\cite[Theorem~3.5.5]{Norris1997}, $p$ is the unique probability vector with $pQ=0$, and $\DKL(q(t)\Vert p)\to 0$ as $t\to\infty$; equivalently $q(t)\to p$ in any norm on $\R^\Omega$. If $Q$ is reducible, the total probability in each closed communicating class is conserved; the limit is obtained by rescaling $p$ on each class to match those conserved masses, and $\DKL(q(t)\Vert p)$ relaxes to the corresponding constrained minimum rather than to zero.
\end{proposition}
\begin{proof}
Differentiating $\DKL(q(t)\Vert p)=\sum_i q_i(t)\log x_i(t)$ and using the master equation~\eqref{eq:master} gives $\frac{\mathrm d}{\mathrm dt}\DKL=\sum_i\dot q_i\log x_i+\sum_i \dot q_i=\sum_i\dot q_i\log x_i$ (since $\sum_i\dot q_i=0$). Substituting $\dot q_j=\sum_i q_iQ_{ij}$ and reorganizing gives
\[
  \frac{\mathrm d}{\mathrm dt}\DKL=\sum_{i\ne j}q_iQ_{ij}\log x_j-\sum_{i\ne j}q_iQ_{ij}\log x_i,
\]
where the second sum uses $\sum_j Q_{ij}=0$. Detailed balance $p_iQ_{ij}=p_jQ_{ji}$ writes $q_iQ_{ij}=p_iQ_{ij}\cdot x_i=p_jQ_{ji}\cdot x_i$. Symmetrizing the indices in unordered pairs $\{i,j\}$ and writing $F_{ij}:=p_iQ_{ij}=p_jQ_{ji}$ yields
\[
  \frac{\mathrm d}{\mathrm dt}\DKL=-\frac{1}{2}\sum_{i\ne j}F_{ij}(x_j-x_i)(\log x_j-\log x_i)\le 0,
\]
because $(x_j-x_i)(\log x_j-\log x_i)\ge 0$ for any pair of positive numbers and $F_{ij}\ge 0$. Convergence to the unique stationary law under irreducibility is the standard ergodic theorem for finite-state continuous-time Markov chains~\cite[Theorem~3.5.5]{Norris1997}.
\end{proof}

The contribution from each unordered pair $\{i,j\}$ vanishes iff $x_i(t)=x_j(t)$, so the equality conditions of~\eqref{eq:KL-derivative} encode the fixed points of the dynamics. In the reducible case the equality is attained as soon as $x_i$ is constant on each communicating class.

Combining this monotonicity with the algebraic identity of Proposition~\ref{thm:fe-kl} gives the monotonicity of recognition free energy along the same dynamics, via a one-line differentiation.

\begin{corollary}[Monotone recognition free energy under reversible Markov dynamics]
\label{cor:monoFE-markov}
Fix $\TR>0$ and let $p$ be the Gibbs measure from Proposition~\ref{thm:fe-kl}. Under the hypotheses of Proposition~\ref{prop:KL-markov}, along any solution of~\eqref{eq:master} in the interior of the simplex,
\[
  \frac{\mathrm d}{\mathrm dt}\,F_{\mathrm{R}}\bigl(q(t)\bigr)\le 0.
\]
\end{corollary}
\begin{proof}
Differentiate the identity $F_{\mathrm{R}}(q(t))-F_{\mathrm{R}}(p)=\TR\,\DKL(q(t)\Vert p)$ from Proposition~\ref{thm:fe-kl} in $t$, then apply Proposition~\ref{prop:KL-markov} to bound the right-hand derivative.
\end{proof}

Thus Proposition~\ref{thm:fe-kl} and Corollary~\ref{cor:fe-min} are static identities on the simplex, while Proposition~\ref{prop:KL-markov} supplies a standard reversible Markov dynamics for which $\DKL(\cdot\Vert p)$ decreases. Corollary~\ref{cor:monoFE-markov} then identifies the same decay as a monotone decrease of $F_{\mathrm{R}}$. Together, these classical statements integrate the RCL-fixed cost vector $X_\omega=J(r_\omega)$ into the standard finite-state free-energy/relative-entropy chain~\cite{Ellis1985,Schnakenberg1976,CoverThomas2006,Touchette2009}. More general paths inherit the same conclusion only when $\DKL(q_t\Vert p)$ is known to be nonincreasing; the physical content then comes from the chosen dynamics, not from the algebraic identity alone, and this kinematic statement should not be conflated with the Boltzmann $H$-theorem for an isolated Hamiltonian gas microstate.

\section{Illustrative example: a three-state reversible system}
\label{sec:example}
We illustrate the results of Sections~\ref{sec:ancestor}--\ref{sec:fe} on a three-state reversible Markov chain---equivalently, a chemical-reaction triangle $A\rightleftharpoons B\rightleftharpoons C\rightleftharpoons A$ at chemical equilibrium---building on the two-state setting recorded in Section~\ref{sec:ledger-bridge}. The example is constructed so that the input ratios $r_\omega$ have a transparent physical meaning: they are detailed-balance odds ratios of the reversible system, and the cost vector $X_\omega=J(r_\omega)$ is then determined by the closed form~\eqref{eq:J-closed}. The variational machinery of Section~\ref{sec:gibbs} produces the unique entropy-maximizing distribution at a chosen mean cost. We then check the free-energy--KL identity of Section~\ref{sec:fe} numerically, illustrate the squared-log bound of Proposition~\ref{prop:divj-basic}, and give a concrete reversible Markov trajectory illustrating the detailed-balance relaxation statement of Proposition~\ref{prop:KL-markov}.

\subsection{Setup: input ratios from detailed balance}
Consider a reversible three-state continuous-time Markov chain on $\Omega=\{1,2,3\}$ with strictly positive stationary law $\pi$ and reference state $1$ (the most populated species, taken as $r_1:=1$). Detailed balance $\pi_i W_{ij}=\pi_j W_{ji}$ identifies, for each non-reference state $\omega\in\{2,3\}$, the dimensionless odds ratio
\[
  r_\omega := \frac{\pi_\omega}{\pi_1}=\frac{W_{1\omega}}{W_{\omega 1}},
\]
exactly as in the two-state setting of Section~\ref{sec:ledger-bridge}. Once ratios are mapped to costs $X_\omega=J(r_\omega)$, Theorem~\ref{thm:types} governs counting asymptotics for arbitrary real costs (hence generic irrational measured ratios); Corollary~\ref{cor:types-exact} records the classical exact-shell refinement in the rational regime. We take the specific values
\[
  r_2 := \tfrac{3-\sqrt5}{2}\approx 0.38197,
  \qquad
  r_3 := 2-\sqrt3\approx 0.26795,
\]
  which are the unique solutions of $J(r_2)=\tfrac12$ and $J(r_3)=1$ in $(0,1)$ under~\eqref{eq:J-closed}, chosen so that the cost vector takes the simple form $X=(0,\tfrac12,1)$; in closed form, $r_2=\varphi^{-2}$ with $\varphi=(1+\sqrt5)/2$ the golden ratio, and $r_3=\tan(\pi/12)=2-\sqrt3$. These choices make the costs rational, so Corollary~\ref{cor:types-exact} (exact integer shell along denominator-exhaustive sequences) applies directly in this subsection. A worked illustration with $r_2=e^{-1}$, $r_3=e^{-3/2}$, for which the exact shell is empty at every $N$ but Theorem~\ref{thm:types} still applies, is given in Section~\ref{sec:soft-shell-numerics}. The chain's own stationary law follows from $\pi_\omega\propto r_\omega$ and normalization:
\[
  \pi=\Bigl(\frac{1}{1+r_2+r_3},\,\frac{r_2}{1+r_2+r_3},\,\frac{r_3}{1+r_2+r_3}\Bigr)
  \approx(0.60609,\ 0.23151,\ 0.16240).
\]
The ratios $r_\omega$ thus encode physically meaningful equilibrium odds: state 1 (the reference) is the most populated, with $r_2\approx 0.38$ giving roughly $38\%$ as much population in state $2$ relative to $1$, and $r_3\approx 0.27$ giving about $27\%$ as much in state $3$. In a reaction-network reading, $r_2$ and $r_3$ are equilibrium constants for the reactions $1\to 2$ and $1\to 3$, both unfavorable (equilibrium constants below unity); the corresponding affinities $\Delta\mu_\omega/k_{\mathrm{B}}T=-\log r_\omega$, evaluated at the convention $k_{\mathrm{B}}T=1$ (so that $\Delta\mu_\omega$ is reported in dimensionless units of $k_{\mathrm{B}}T$), are $\Delta\mu_2\approx 0.962$ and $\Delta\mu_3\approx 1.317$.

\subsection{Gibbs distribution at fixed mean cost}
With the cost vector $X=(0,\tfrac12,1)$ fixed by the input ratios, fix a constraint value $\Ebar=0.45$. This value lies strictly between $\min X=0$ and the uniform mean $\langle X\rangle_{p^{(0)}}=\tfrac13(0+\tfrac12+1)=\tfrac12$, so Proposition~\ref{prop:inv-beta} guarantees a unique solution $\beta>0$ to the moment-matching equation
\begin{equation}
  \frac{0\cdot e^{0}+(1/2)e^{-\beta/2}+1\cdot e^{-\beta}}{1+e^{-\beta/2}+e^{-\beta}}=\Ebar,
  \label{eq:three-state-beta}
\end{equation}
which is $\langle X\rangle_{p^{(\beta)}}=\Ebar$ with $X=(0,\tfrac12,1)$ written out. Numerical solution gives
\[
  \beta\approx 0.30113,\qquad \TR=\beta^{-1}\approx 3.3208,\qquad
  Z\approx 2.6002,
\]
and Gibbs weights
\[
  p\approx (0.38459,\ 0.33083,\ 0.28459),
\]
monotonically decreasing in the cost: the reference state has the largest weight, and the highest-cost state the smallest. For comparison, the uniform reference distribution $q=(\tfrac13,\tfrac13,\tfrac13)$ corresponds to $\beta=0$ and the higher mean cost $\langle X\rangle_q=\tfrac12$.

The Gibbs distribution $p$ is the maximum-entropy distribution among all probability laws on $\Omega$ with the prescribed mean cost $\Ebar=0.45$. It need not coincide with the chain's own stationary distribution $\pi\approx(0.6061,0.2315,0.1624)$ from the previous subsection: $\pi$ has its own realized mean cost $\langle X\rangle_\pi=(r_2/2+r_3)/(1+r_2+r_3)\approx 0.2782$, distinct from $\Ebar$. The two distributions coincide only when the chain's stationary distribution is itself a Gibbs measure for the cost vector $X$ at the matching inverse temperature: $\pi=p^{(\beta)}$ for some $\beta>0$ if and only if the chosen constraint value $\Ebar$ equals $\langle X\rangle_\pi$, and the matching $\beta$ is the unique root from Proposition~\ref{prop:inv-beta}. The Gibbs distribution of Section~\ref{sec:gibbs} is a thermodynamic object selected by entropy maximization at fixed mean cost, while the chain's stationary law is a kinematic object determined by its rate matrix; the two come from compatible but distinct constructions.

\subsection{Numerical values of the divergences}
\label{sec:example-numerics}
With $p$ and $q$ specified, the functionals of Sections~\ref{sec:divj}--\ref{sec:fe} take the closed-form numerical values shown in Table~\ref{tab:three}. The free-energy--KL identity of Proposition~\ref{thm:fe-kl} is an algebraic equality and the agreement of the rows $F_{\mathrm{R}}(q)-F_{\mathrm{R}}(p)$ and $\TR\,\DKL(q\Vert p)$ in Table~\ref{tab:three} only documents that the numerical implementation is consistent with the closed-form identity. The interesting comparison is between $\DivJ$ and the squared-log surrogate of Lemma~\ref{lem:ancestor}: the inequality $\DivJ(q\Vert p)\ge \tfrac12\sum_\omega p_\omega(\log(q_\omega/p_\omega))^2$ from Proposition~\ref{prop:divj-basic}\,(iii) is satisfied with the small positive margin $0.00751-0.00750\approx 1\times 10^{-5}$ visible in the last two rows, as expected near coincidence. The next subsection turns to a genuinely dynamical illustration of the detailed-balance Lyapunov statement.

\begin{table}[htbp]
  \centering
  \begin{tabular}{lr}
    \hline
    Quantity & Value (rounded)\\\hline
    $H(p)$ & $1.0911$\\
    $H(q)$ & $\log 3\approx 1.0986$\\
    $\DKL(q\Vert p)$ & $0.00754$\\
    $F_{\mathrm{R}}(p)$ & $-3.1733$\\
    $F_{\mathrm{R}}(q)$ & $-3.1483$\\
    $F_{\mathrm{R}}(q)-F_{\mathrm{R}}(p)=\TR\,\DKL(q\Vert p)$ & $0.02505$\\
    $\DivJ(q\Vert p)$ & $0.00751$\\
    $\tfrac12\sum_\omega p_\omega\bigl(\log(q_\omega/p_\omega)\bigr)^2$ & $0.00750$\\\hline
    \multicolumn{2}{l}{\footnotesize $p$ denotes the Gibbs measure at mean cost $\Ebar=0.45$;}\\
    \multicolumn{2}{l}{\footnotesize $q=(\tfrac13,\tfrac13,\tfrac13)$ is the uniform reference distribution.}
  \end{tabular}
  \caption{Three-state example with $\Ebar=0.45$. The row $F_{\mathrm{R}}(q)-F_{\mathrm{R}}(p)=\TR\DKL(q\Vert p)$ records the algebraic identity of Proposition~\ref{thm:fe-kl}; the last two rows illustrate the squared-log lower bound of Proposition~\ref{prop:divj-basic}\,(iii).}
  \label{tab:three}
\end{table}

\subsection{Detailed-balance relaxation trajectory}
\label{sec:example-relaxation}
To put Proposition~\ref{prop:KL-markov} on this example, we equip the three states with an irreducible continuous-time Markov chain that has $p$ as its unique stationary law and is reversible with respect to it. The Metropolis-on-a-graph rates~\cite{LevinPeresWilmer2017},
\[
  Q_{ij}=\min\!\Bigl(1,\tfrac{p_j}{p_i}\Bigr)\quad (i\neq j),\qquad
  Q_{ii}=-\sum_{j\neq i}Q_{ij},
\]
satisfy detailed balance $p_i Q_{ij}=p_j Q_{ji}$ by construction, with row sums zero. With the Gibbs reference $p\approx(0.385,0.331,0.285)$ from Section~\ref{sec:example-numerics}, the off-diagonal block evaluates to
\[
  Q_{12}=0.860,\;Q_{13}=0.740,\;Q_{21}=1.000,\;Q_{23}=0.860,\;Q_{31}=Q_{32}=1.000.
\]
Starting from the uniform initial law $q(0)=(\tfrac13,\tfrac13,\tfrac13)$ and integrating $\dot{q}=qQ$ via $q(t)=q(0)\exp(tQ)$ gives the trajectory in Table~\ref{tab:relaxation}.

\begin{table}[htbp]
  \centering
  \begin{tabular}{rrrrrr}
    \hline
    $t$ & $q_1(t)$ & $q_2(t)$ & $q_3(t)$ & $\DKL(q(t)\Vert p)$ & $F_{\mathrm{R}}(q(t))$\\\hline
    $0.0$ & $0.3333$ & $0.3333$ & $0.3333$ & $7.54\times 10^{-3}$ & $-3.14830$\\
    $0.5$ & $0.3706$ & $0.3323$ & $0.2970$ & $5.29\times 10^{-4}$ & $-3.17159$\\
    $1.0$ & $0.3808$ & $0.3314$ & $0.2878$ & $3.74\times 10^{-5}$ & $-3.17322$\\
    $2.0$ & $0.3843$ & $0.3309$ & $0.2848$ & $1.94\times 10^{-7}$ & $-3.17335$\\
    $5.0$ & $0.3846$ & $0.3308$ & $0.2846$ & $<10^{-12}$ & $-3.17335$\\\hline
  \end{tabular}
  \caption{Relaxation under the Metropolis chain reversible with respect to $p$, started from the uniform law. Both $\DKL(q(t)\Vert p)$ and $F_{\mathrm{R}}(q(t))$ are monotonically decreasing in $t$, as predicted by Proposition~\ref{prop:KL-markov} and Corollary~\ref{cor:monoFE-markov}; $q(t)\to p$ exponentially, and $F_{\mathrm{R}}(q(t))$ approaches the equilibrium value $F_{\mathrm{R}}(p)\approx -3.1733$.}
  \label{tab:relaxation}
\end{table}

The trajectory is a direct numerical realization of Proposition~\ref{prop:KL-markov}: $\DKL$ falls from $7.54\times 10^{-3}$ to below $10^{-12}$ over five (arbitrary) Markov-time units, and $F_{\mathrm{R}}(q(t))-F_{\mathrm{R}}(p)=\TR\DKL(q(t)\Vert p)$ tracks the same exponential decay rescaled by $\TR=1/\beta\approx 3.32$. The leading relaxation rate is set by the spectral gap of $Q$, and it is the standard exponential approach to canonical equilibrium for the chosen rates; what changes from one model to another is only the cost vector $X$ that enters $p$ via the variational principle.

\subsection{\texorpdfstring{Comparison with the Tsallis $q$-exponential alternative}{Comparison with the Tsallis q-exponential alternative}}
\label{sec:tsallis-compare}
The numerical agreement above checks an algebraic identity. To turn the example into a non-trivial comparison, we juxtapose the RCL Gibbs distribution with the Tsallis $q$-exponential obtained under an alternative entropy functional on the same cost data. To avoid clashing with the symbol $q$ already used for probability distributions, we write $q_T$ for the Tsallis deformation index throughout this subsection.

The Tsallis program~\cite{Tsallis1988,Naudts2011} replaces Shannon entropy by the $q_T$-deformed entropy
\[
  S_{q_T}(p):=\frac{1}{q_T-1}\Bigl(1-\sum_{\omega\in\Omega}p_\omega^{q_T}\Bigr),\qquad q_T>0,
\]
which reduces to $H$ as $q_T\to 1$. Maximizing $S_{q_T}$ at fixed $\langle X\rangle_p=\Ebar$ on the simplex (constraint-1 formulation; see~\cite{Tsallis1988,Naudts2011}) yields
\[
  p^{(\beta_{q_T})}_\omega \propto \bigl[1-(1-q_T)\beta_{q_T} X_\omega\bigr]_+^{1/(1-q_T)},
\]
with $\beta_{q_T}$ adjusted so that $\sum_\omega p^{(\beta_{q_T})}_\omega X_\omega=\Ebar$. For $q_T=1$ this reduces to the Gibbs form; for $q_T\neq 1$ the tails are power-law rather than exponential. Table~\ref{tab:tsallis-compare} fits $\beta_{q_T}$ to the same constraint $\Ebar=0.45$ as Table~\ref{tab:three} and reports the $\ell^1$ distance to the RCL Gibbs measure $p^{\mathrm{RCL}}\approx(0.385,0.331,0.285)$.

\paragraph{Scope of the comparison.}
This juxtaposition is a controlled deformation of the Gibbs branch within the fixed RCL cost vector $X=(0,\tfrac12,1)$ and the simple-mean (constraint-1) formulation used throughout this paper. It is not a comparison against the full Tsallis statistical-mechanical framework, where the natural cost choice and the escort-mean (constraint-3) formulation may differ from the present setup; conclusions drawn from Table~\ref{tab:tsallis-compare} should not be read as general statements about Tsallis non-extensive statistical mechanics. Its purpose is only to illustrate that, on the same finite cost data and mean-cost constraint, deforming the entropy functional within the $q_T$-family yields distributions that differ from the Gibbs branch in a controlled way and become indistinguishable as $q_T\to 1$.

\begin{table}[htbp]
  \centering
  \begin{tabular}{lrrrrrr}
    \hline
    Distribution & Tsallis index $q_T$ & $\beta_{q_T}$ & $p_1$ & $p_2$ & $p_3$ & $\lVert\cdot-p^{\mathrm{RCL}}\rVert_1$\\\hline
    RCL Gibbs & $1.00$ & $0.3011$ & $0.3846$ & $0.3308$ & $0.2846$ & $0$\\
    Tsallis $q_T$-exp & $0.10$ & $0.2646$ & $0.3835$ & $0.3331$ & $0.2835$ & $0.0045$\\
    Tsallis $q_T$-exp & $0.50$ & $0.2801$ & $0.3840$ & $0.3321$ & $0.2840$ & $0.0025$\\
    Tsallis $q_T$-exp & $1.50$ & $0.3243$ & $0.3852$ & $0.3296$ & $0.2852$ & $0.0025$\\
    Tsallis $q_T$-exp & $2.50$ & $0.3780$ & $0.3864$ & $0.3272$ & $0.2864$ & $0.0073$\\\hline
  \end{tabular}
  \caption{Three-state example with $X=(0,\tfrac12,1)$ and $\Ebar=0.45$. Tsallis $q_T$-exponential distributions match the mean-cost constraint at different shape parameters $\beta_{q_T}$; the resulting weights differ from the RCL Gibbs measure by an $\ell^1$ amount that increases as the deformation moves away from the Gibbs case $q_T=1$. The gap is small at $\Ebar=0.45$ because the constraint sits close to the uniform-cost value $0.5$, where the high-temperature limit flattens all $q_T$-exponential families to the uniform measure.}
  \label{tab:tsallis-compare}
\end{table}

At the matched constraint $\Ebar=0.45$, the Tsallis entropy $S_{q_T}(p^{(\beta_{q_T})})$ and the Helmholtz analogue $F^{(q_T)}=\langle X\rangle_p-T_{q_T}S_{q_T}(p)$ with $T_{q_T}:=1/\beta_{q_T}$ evaluate to $S_{q_T=0.5}=1.4578$ and $F^{(0.5)}\approx -4.7543$; $S_{q_T=1}=H(p^{\mathrm{RCL}})=1.0911$ and $F_{\mathrm{R}}(p^{\mathrm{RCL}})\approx -3.1737$; $S_{q_T=1.5}=0.8388$ and $F^{(1.5)}\approx -2.1365$; and $S_{q_T=2.5}=0.5347$ and $F^{(2.5)}\approx -0.9646$. These values are computed analytically from the Table~\ref{tab:tsallis-compare} weights and the fitted $\beta_{q_T}$ entries.

For $q_T<1$ the population shifts slightly toward the median-cost state and away from the extremes within this fixed-cost-and-constraint setup; for $q_T>1$ the opposite occurs. The separation grows as $\Ebar$ moves away from the uniform value: at $\Ebar=0.20$, $\lVert p^{\mathrm{Tsallis}}_{q_T=0.5}-p^{\mathrm{RCL}}\rVert_1\approx 0.087$ and $\lVert p^{\mathrm{Tsallis}}_{q_T=2.5}-p^{\mathrm{RCL}}\rVert_1\approx 0.127$, an order of magnitude larger than at $\Ebar=0.45$. Within this deformation, the same mean cost is therefore compatible with a one-parameter family of distributions; concretely, the observables that resolve the families are those sensitive to more than the mean of $X$, including the variance $\mathrm{Var}_p(X)$, the high-cost tail probability $p_3$, and any moment $\langle X^k\rangle_p$ for $k\ge 2$. The comparison itself is not an empirical validation of the RCL, and it is not a comparison against the full Tsallis non-extensive statistical-mechanical framework: it only demonstrates that, within a fixed cost vector and simple-mean constraint, deforming the entropy functional within the $q_T$-family yields finite-state distributions distinguishable from the Gibbs branch at the fluctuation level rather than at the mean.

\subsection{Sample-size requirements at fixed RCL ground truth (power calculation)}
\label{sec:discriminating-observables}
This subsection is a power calculation, not a falsifiability test of the RCL: the RCL Gibbs law $p^{\mathrm{RCL}}$ is taken as the generating distribution, and each alternative framework is asked how many effectively independent observations would be needed to resolve its predicted tail weight from that of $p^{\mathrm{RCL}}$ at nominal $2\sigma$ on a single coordinate. An actual falsifiability test would require an experimental three-state system with independently measured rate ratios and stationary populations.

Throughout this subsection and in Table~\ref{tab:discriminate}, the quantity $N_{2\sigma}$ denotes the \emph{effective} number of independent Bernoulli trials equivalent, for inferring the frequency $\hat p_3$, to the Wald two-sided $z=2$ sample size after correcting for temporal autocorrelation in Markovian trajectories (thinning, block averaging, integrated autocorrelation time, or related standard procedures).

An $\ell^1$ distance between two candidate distributions is a useful summary statistic, but it does not directly tell an experimenter whether the two distributions can be distinguished at a specified confidence level. We therefore convert the model gap into a sample-size requirement on the most informative single coordinate.

The Tsallis comparison of Section~\ref{sec:tsallis-compare} reports an $\ell^1$ distance, but does not specify which physical observable would discriminate between the alternatives or how many samples would be needed to do so. The natural finite-state observable is the population frequency $\hat p_\omega$ on each state, accumulated from $N$ independent observations. Among the three frequencies, the highest-cost weight $p_3$ is the most directly informative, since it is the most sensitive to the cost choice in the Gibbs construction.

To put cost frameworks on equal footing, we adopt the following protocol. Take the RCL Gibbs distribution $p^{\mathrm{RCL}}$ at $\Ebar=0.45$ as the generating alternative in a model-discrimination power calculation (Section~\ref{sec:example-numerics}). Each candidate cost framework is asked to explain that physical state by fitting its inverse temperature to the empirical mean of its own cost computed at $p^{\mathrm{RCL}}$. The four candidates on the same ratio data $(r_1,r_2,r_3)=(1,\varphi^{-2},\tan(\pi/12))$ are: the RCL cost $X^{\mathrm{RCL}}_\omega=J(r_\omega)=(0,\tfrac12,1)$ (Section~\ref{sec:example}); the squared-log surrogate $X^{\mathrm{SQ}}_\omega=\tfrac12(\log r_\omega)^2\approx(0,0.463,0.867)$ from Lemma~\ref{lem:ancestor}; the affinity-as-energy assignment $X^{\mathrm{A}}_\omega=-\log r_\omega\approx(0,0.962,1.317)$, which treats the reference-oriented log-ratio directly as the energy-like coordinate (the canonical-ensemble version of the stochastic-thermodynamic comparison in Section~\ref{sec:stoch-thermo}); and the Tsallis $q_T$-exponential at the same RCL cost vector and constraint value (Table~\ref{tab:tsallis-compare}).

Table~\ref{tab:discriminate} reports the predicted $p_3$ for each framework together with $N_{2\sigma}$ as defined in the table caption, required to distinguish that prediction from $p^{\mathrm{RCL}}_3=0.2846$ at the $2\sigma$ level on a single coordinate of the empirical type. This estimate is a deliberately simple power calculation: it uses a prespecified single-coordinate statistic and the binomial standard error $\sigma^2_{\hat p_3}=p_3(1-p_3)/N$ only after $N$ has been reduced to an effective independent-sample count as in the opening paragraph. A multi-coordinate goodness-of-fit test would require the corresponding likelihood-ratio or multiple-comparison calibration. The affinity-as-energy alternative is the most readily distinguishable: at $\Ebar=0.45$ it predicts $p_3\approx 0.295$ versus the RCL value $0.285$, a binomial-distinguishable difference at $N\approx 8\!\times\!10^{3}$---in principle accessible in systems that provide on the order of $10^4$ effectively independent state observations after autocorrelation correction. The squared-log surrogate and the Tsallis $q_T=2.5$ case require $N\sim 10^{5}$--$10^{6}$, which is more demanding and would be realistic only in high-throughput or highly repeatable assays. The Tsallis case at $q_T=0.5$, very close to the Gibbs branch $q_T=1$, is essentially indistinguishable.

\begin{table}[htbp]
  \centering
  \footnotesize
  \begin{tabular}{l c c r r r}
    \hline
    Framework & Cost vector $X_\omega$ & Fitted $\beta$ & $p_3$ & $\Delta p_3$ & $N_{2\sigma}$\\\hline
    RCL Gibbs & $(0,0.500,1.000)$ & $0.301$ & $0.2846$ & $0$ & ---\\
    Squared-log surrogate & $(0,0.463,0.867)$ & $0.346$ & $0.2858$ & $+0.0012$ & $5.7\!\times\!10^{5}$\\
    Affinity-as-energy & $(0,0.962,1.317)$ & $0.210$ & $0.2945$ & $+0.0099$ & $8.3\!\times\!10^{3}$\\
    Tsallis $q_T=0.5$ & $(0,0.500,1.000)$ & $0.280$ & $0.2840$ & $-0.0006$ & $2.1\!\times\!10^{6}$\\
    Tsallis $q_T=2.5$ & $(0,0.500,1.000)$ & $0.378$ & $0.2864$ & $+0.0018$ & $2.4\!\times\!10^{5}$\\\hline
  \end{tabular}
  \caption{Three-state example at $\Ebar=0.45$. Each candidate framework fits its inverse temperature $\beta$ to match the mean of its own cost computed at $p^{\mathrm{RCL}}$, and predicts the heavy-state weight $p_3$. The column $N_{2\sigma}$ is the Wald two-sided $z=2$ binomial sample-size formula $N_{2\sigma}\approx 4p_3(1-p_3)/(\Delta p_3)^2$, i.e.\ the number of \emph{effectively independent} observations of state~$3$ needed to resolve the shift $\Delta p_3$ from $p^{\mathrm{RCL}}_3=0.2846$ at nominal $2\sigma$ on $\hat p_3$ alone, after reducing raw trajectory length to an equivalent independent-sample count when data are Markov-correlated (cf.\ the opening of Section~\ref{sec:discriminating-observables}).}
  \label{tab:discriminate}
\end{table}

The discriminating power increases sharply as $\Ebar$ moves away from the high-temperature limit $\Ebar\to\tfrac12$. At $\Ebar=0.20$ (recognition temperature $\TR\approx 0.47$, deeper into the cost-penalized regime), the same protocol gives $p^{\mathrm{RCL}}_3\approx 0.082$ and predicted differences $\Delta p_3\approx 0.035$ for the affinity-as-energy alternative, $0.032$ for Tsallis $q_T=2.5$, and $0.022$ for Tsallis $q_T=0.5$, requiring binomial samples of only $N_{2\sigma}\approx 250$, $300$, and $640$ respectively---within reach of moderate-throughput dwell-time histograms or equilibrium-population assays in favorable systems~\cite{Seifert2012,RaoEsposito2016}. The squared-log surrogate, which agrees with the RCL cost only to second order at $r=1$, becomes distinguishable at $N_{2\sigma}\approx 1.3\!\times\!10^{4}$ in this regime.

The qualitative content is therefore that the present finite-state framework makes quantitatively distinct predictions from each of the natural alternatives---squared-log, affinity-as-energy, and Tsallis---on the same input ratios, and that the most directly competing framework (affinity-as-energy, which one would naturally adopt by treating the dimensionless edge affinity $\log r_\omega$ as an energy) is distinguishable at sample sizes that are, in principle, accessible after autocorrelation correction over the illustrated parameter range, becoming sharply distinguishable away from the uniform-cost limit. The discriminating observable need not be sophisticated: a single coordinate of the empirical type vector suffices. These estimates measure how distinguishable the RCL prediction is from each alternative at fixed mean cost, conditional on the RCL Gibbs law being the generating distribution. They are not a falsifiability test of the RCL itself; an actual test would require an experimental three-state system with independently measured rate ratios and stationary populations, so that the RCL, affinity-as-energy, and Tsallis predictions could each be compared against measured data.

\subsection{Dependence on the constraint}
Varying $\Ebar$ within its admissible range $(0,\tfrac12)$ while keeping the cost vector $X=(0,\tfrac12,1)$ fixed gives a more global view. As $\Ebar$ increases from $0$ toward the uniform mean $\tfrac12$, the canonical parameter $\beta$ from Proposition~\ref{prop:inv-beta} decreases from $+\infty$ to $0^+$, the recognition temperature $\TR=1/\beta$ rises from $0$ to $+\infty$, and the Gibbs measure deforms continuously from the point mass on the zero-cost outcome to the uniform measure. Table~\ref{tab:Ebar-sweep} lists the canonical parameters at three sample points approaching the midpoint $\Ebar=\tfrac12$ from below; the rapid divergence of $\TR$ and the flattening of the Gibbs weights are visible from the numerical entries.

Figure~\ref{fig:gibbs_weights} displays this dependence graphically. The weight $p_1$ on the zero-cost outcome decays monotonically as $\Ebar$ increases, while $p_3$ on the highest-cost outcome rises monotonically; $p_2$ traces a milder curve between the two. All three curves converge to $\tfrac13$ as $\Ebar\to\tfrac12$, the high-temperature limit in which the cost contribution to the Gibbs weights becomes negligible. This deformation is the finite-state shadow of the standard high-temperature limit of the canonical ensemble.

\begin{table}[htbp]
  \centering
  \begin{tabular}{lrrrr}
    \hline
    $\langle X\rangle$ (target) & $\beta$ & $\TR=1/\beta$ & $Z$ & $p_1$ (rounded)\\\hline
    $0.35$ & $0.932$ & $1.073$ & $2.021$ & $0.495$\\
    $0.45$ & $0.301$ & $3.321$ & $2.600$ & $0.385$\\
    $0.48$ & $0.120$ & $8.328$ & $2.829$ & $0.354$\\\hline
  \end{tabular}
  \caption{Same three-state costs $X=(0,\tfrac12,1)$ as Table~\ref{tab:three}; canonical parameters as the target mean cost $\langle X\rangle$ (the constraint value denoted $\Ebar$ elsewhere) varies (numerical root of $\langle X\rangle_{p^{(\beta)}}=\Ebar$).}
  \label{tab:Ebar-sweep}
\end{table}

\begin{figure}[htbp]
  \centering
  \includegraphics[width=0.6\textwidth]{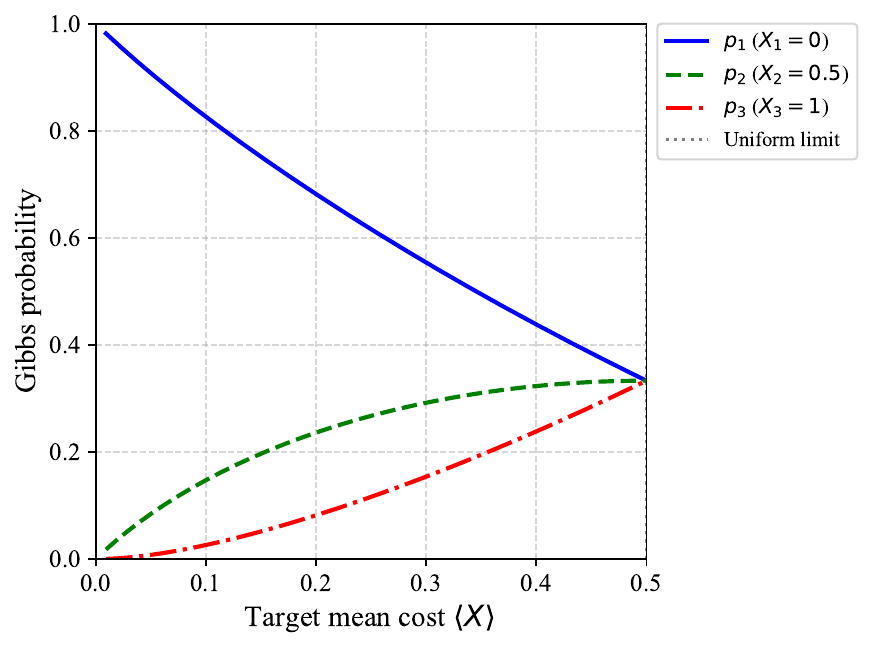}
  \caption{Gibbs weights $p_k$ as a function of the target mean cost $\langle X\rangle$ for the three-state example with costs $X=(0,0.5,1)$, plotted over the positive-temperature range $\Ebar\in(0,\tfrac12)$ in which $\beta>0$ (Proposition~\ref{prop:inv-beta}); $\Ebar$ denotes the same constraint as the column $\langle X\rangle$ (target) in Table~\ref{tab:Ebar-sweep}. The dotted vertical line marks the high-temperature limit $\Ebar\to\tfrac12$ at which $\beta\to 0^+$ and the distribution approaches the uniform measure.}
  \label{fig:gibbs_weights}
\end{figure}

\FloatBarrier
\section{Discussion}
\label{sec:discussion}

The derivation keeps three steps explicit and separate: fix the scalar cost from the RCL, recover Shannon entropy from multinomial counting, and apply convex optimization to obtain Gibbs weights. Once $X_\omega$ is fixed, Section~\ref{sec:gibbs} reproduces the maximum-entropy characterization of the Gibbs measure and Section~\ref{sec:fe} the free-energy--Kullback--Leibler identity for the canonical exponential family on a finite space.

\subsection{Logical summary of the derivation}
\label{sec:logical-summary}
Table~\ref{tab:chain} lists the nine steps of the derivation in order; each row is a consequence of those above together with a single new ingredient (calibration, counting, or convex analysis).

\begin{table}[htbp]
  \centering
  \footnotesize
  \setlength{\tabcolsep}{4pt}
  \begin{tabular}{c p{0.42\textwidth} p{0.40\textwidth}}
    \hline
    Step & Result & Reference\\\hline
    1 & (A1)--(A3) $\Rightarrow$ $J(x)=\tfrac12(x+x^{-1})-1$ & Theorem~\ref{thm:wz-uniqueness}\\
    2 & $J(x)\ge (\log x)^2/2$ & Lemma~\ref{lem:ancestor}\\
    3 & $\DivJ$ properties; half-Neyman $\chi^2$ form & Proposition~\ref{prop:divj-basic}\\
    4 & Shannon entropy from constrained multinomials & Theorem~\ref{thm:types}; Corollary~\ref{cor:types-exact}; non-asymptotic bound, Proposition~\ref{prop:stirling-rate}; multi-constraint soft shell, Proposition~\ref{prop:soft-shell-multi}\\
    5 & Existence and uniqueness of $\beta$ ($\TR=1/\beta$ for $\beta>0$) & Lemma~\ref{lem:logZ}, Proposition~\ref{prop:inv-beta}\\
    6 & Finite-state Gibbs variational principle & Proposition~\ref{thm:maxent}\\
    7 & $F_{\mathrm{R}}(q)-F_{\mathrm{R}}(p)=\TR\,\DKL(q\Vert p)$ & Proposition~\ref{thm:fe-kl}\\
    8 & Gibbs minimizes $F_{\mathrm{R}}$ on the simplex & Corollary~\ref{cor:fe-min}\\
    9 & $\DKL$ nonincreasing under detailed balance & Proposition~\ref{prop:KL-markov}, Corollary~\ref{cor:monoFE-markov}\\\hline
  \end{tabular}
  \caption{Step-by-step derivation chain. Origins of each result and the contribution categories (N1)--(N3) are stated in Section~\ref{sec:axioms}.}
  \label{tab:chain}
\end{table}

\subsection{Relation to Jaynes and standard statistical mechanics}
\label{subsec:jaynes}
In Jaynes's maximum-entropy program~\cite{Jaynes1957,Jaynes1957II,PresseGhoshLeeDill2013}, a scalar cost---typically a Hamiltonian or phenomenological energy---is chosen, its mean is fixed on the simplex, and Shannon entropy is maximized, yielding the Gibbs weights and an inverse temperature as Lagrange multiplier. The remaining modeling freedom lies in the choice of cost functional and constraint value; multinomial counting, Stirling entropy, and the convex duality identities are shared with the present work.

The present work differs in one place: the cost is selected upstream by the RCL together with~(A2), rather than chosen as an external Hamiltonian. Once dimensionless input ratios $\{r_\omega\}$ are specified, $X_\omega=J(r_\omega)$ has the closed form~\eqref{eq:J-closed}. Modeling freedom is relocated from the choice of cost functional to the choice of ratio assignment, not eliminated: specifying $\{r_\omega\}$ is itself a modeling commitment of comparable difficulty to specifying a Hamiltonian, and the Tsallis comparison of Section~\ref{sec:tsallis-compare} shows that even with $\{r_\omega\}$ fixed there is residual freedom in choosing the composition law.

The Shore--Johnson axiomatization~\cite{ShoreJohnson1980,Caticha2012} is complementary. Shore and Johnson fix Shannon entropy as the unique self-consistent inference rule under linear constraints; here Shannon entropy is taken as given, emerging from multinomial counting in Section~\ref{sec:counting}, and the focus is on which cost the entropy is maximized against. In one slogan: Shore--Johnson constrains the entropy; the RCL constrains the cost. Caticha's entropic inference~\cite{Caticha2004,Caticha2012} extends the Shore--Johnson program into a broader epistemic setting on which the RCL is silent. A natural continuation, which we leave open, is to ask whether the Shore--Johnson axioms restated in ratio coordinates would force, or merely permit, the d'Alembert composition law~\eqref{eq:composition}; we return to this in Section~\ref{sec:conclusion}.

From the large-deviation perspective, Sanov's theorem and its conditional formulations describe how empirical measures concentrate under linear constraints in the large-$N$ limit~\cite{Sanov1957,DemboZeitouni2010,CsiszarKorner1981}, with the Gibbs measure as the unique minimizer of the rate function. Theorem~\ref{thm:types} is the finite-$K$ method-of-types instance of this picture~\cite{CoverThomas2006}, with rate of convergence controlled by Proposition~\ref{prop:stirling-rate}; see~\cite{Touchette2009,Ellis1985} for the broader framework. The information-geometric viewpoint of Amari~\cite{Amari1985,AmariNagaoka2000} describes $\{p^{(\beta)}\}_{\beta\in\R}$ as a one-dimensional exponential family with Fisher metric induced by $\log Z$ and Proposition~\ref{thm:fe-kl} as the Legendre--Fenchel duality between natural and expectation parameters; this geometry is the natural language for extensions to higher-dimensional sufficient statistics, but is not pursued here.

\paragraph{\texorpdfstring{Neighboring literatures on the $\cosh$-kernel and on multi-variable functional equations.}{Neighboring literatures on the cosh-kernel and on multi-variable functional equations.}} The closed cost $J(e^t)=\cosh t-1$ is a hyperbolic-cosine kernel in the log-ratio coordinate $t=\log r$, a form that appears in several neighboring information-theoretic and information-geometric literatures. The Itakura--Saito divergence~\cite{ItakuraSaito1968} and its symmetrized variants on the positive cone $\R_{>0}$ are ratio-based contrasts---built directly from $r=q/p$---in which the symmetric (even) component is of the same hyperbolic-cosine type as $J$, while the antisymmetric (odd) component is the log-ratio that also appears in the affinity $\log r$ of Section~\ref{sec:stoch-thermo}; in the symmetric case both Itakura--Saito and the present $J$ vanish quadratically at $r=1$ and grow superquadratically in $|\log r|$. On positive-definite matrices, the symmetrized Bregman or log-determinant Stein divergences~\cite{ChebbiMoakher2012,Sra2012} exhibit a related $\cosh$-of-log-eigenvalue structure as the natural even part of the asymmetric log-determinant divergence. The Gibbs construction here recovers the one-dimensional exponential family with cumulant generator $\log Z(\beta)$; whether the RCL Gibbs measure carries an additional information-geometric interpretation beyond this family---for example a non-trivial dually-flat structure on the positive cone induced by $\DivJ$---is not addressed in this paper. For the multi-variable functional-equation theory used in Theorem~\ref{thm:wz-uniqueness} (the d'Alembert equation on $\R^2$ in log-coordinates) the standard reference beyond Acz\'el~\cite{Aczel1966} and Kuczma--Choczewski--Ger~\cite{Kuczma2009} is Acz\'el and Dhombres~\cite{AczelDhombres1989}.

\subsection{\texorpdfstring{Connection to stochastic thermodynamics: $J(r)$ as one symmetric summary of the edge data}{Connection to stochastic thermodynamics: J(r) as one symmetric summary of the edge data}}
\label{sec:stoch-thermo}
The choice to enter the ratios $\{r_\omega\}$ through the cost $X_\omega=J(r_\omega)$, rather than through the affinity $\log r_\omega$, has a direct interpretation in the language of stochastic thermodynamics~\cite{Seifert2012,EspositoVanDenBroeck2010}. Consider a reversible continuous-time Markov chain on $\Omega$ with positive transition rates $W_{ij}$ ($i\ne j$) satisfying detailed balance $\pi_iW_{ij}=\pi_jW_{ji}$ with respect to a strictly positive stationary law $\pi$~\cite{LevinPeresWilmer2017}. Each edge $i\!\leftrightarrow\!j$ then carries a single dimensionless ratio
\[
  r_{ij}:=W_{ij}/W_{ji}=\pi_j/\pi_i,\qquad r_{ji}=r_{ij}^{-1}.
\]
This scalar admits two natural summaries on the same edge data. The rate-ratio contribution to the stochastic-thermodynamic force is
\begin{equation}
  A^{\mathrm{rate}}_{ij}:=\log r_{ij}=-A^{\mathrm{rate}}_{ji},
  \label{eq:affinity}
\end{equation}
which is antisymmetric under edge reversal. At a nonequilibrium distribution $q$, the corresponding Schnakenberg force is
\[
  \mathcal A_{ij}(q):=\log\frac{q_iW_{ij}}{q_jW_{ji}}
  =\log\frac{q_i}{q_j}+A^{\mathrm{rate}}_{ij}.
\]
The edge entropy-production contribution is
\[
  (q_iW_{ij}-q_jW_{ji})\mathcal A_{ij}(q)\ge 0,
\]
and the full force $\mathcal A_{ij}(q)$, not the bare rate-ratio contribution $A^{\mathrm{rate}}_{ij}$, vanishes at detailed-balance equilibrium $q=\pi$~\cite{Schnakenberg1976,Seifert2012,EspositoVanDenBroeck2010}. One symmetric summary of the same edge-ratio data is the RCL cost
\begin{equation}
  J(r_{ij})=\cosh A^{\mathrm{rate}}_{ij}-1=J(r_{ji})
  \label{eq:J-vs-affinity}
\end{equation}
(\textit{Any even function $f$ with $f(0)=0$ furnishes an alternative symmetric summary of the same edge ratio, for instance via $f(\log r_{ij})$ or $f(|A^{\mathrm{rate}}_{ij}|)$.})
It is one natural choice among them, selected by the RCL axioms (S1)--(S3) and the degree-two closure of Section~\ref{sec:rcl}: it is even under edge reversal, depends only on $|A^{\mathrm{rate}}_{ij}|$, and vanishes precisely when the bare edge ratio is unity, $r_{ij}=1$, i.e.\ when the two oriented rates are equal before weighting by the stationary distribution. This is distinct from detailed-balance equilibrium, for which the full force $\mathcal A_{ij}(q)$ vanishes at $q=\pi$ even when $r_{ij}\neq 1$. Canonicity is not claimed, and proving that $J$ is the unique such functional compatible with detailed-balance entropy-production inequalities or current-fluctuation symmetries would require additional axioms not stated here.

This parity decomposition gives the cost $X_\omega=J(r_\omega)$ a recognizable physical role: it is a symmetric, equilibrium-oriented summary of the same rate-ratio data whose antisymmetric component enters stochastic thermodynamics through forces and currents. The Gibbs construction of Sections~\ref{sec:gibbs}--\ref{sec:fe} uses only this symmetric content $J$, consistent with its restriction to finite-state equilibrium canonical-ensemble objects; the decomposition itself, however, is not unique (any even function of $\log r$ with $f(0)=0$ furnishes a symmetric summary), and the same latitude applies when reading $J$ against the stochastic-thermodynamic literature. Cycle affinities, broken-detailed-balance currents, and trajectory-level fluctuation relations~\cite{PolettiniEsposito2014,RaoEsposito2016,EspositoVanDenBroeck2010} would require retaining the antisymmetric/current component as additional structure, in line with the information-geometric thermodynamic interpretation of Ito~\cite{Ito2018} and the thermodynamic uncertainty relations bounding precision against dissipation~\cite{BaratoSeifert2015}. We do not pursue the corresponding non-equilibrium extension here; the present finite-state framework keeps only the symmetric content of each edge and is therefore restricted to detailed-balance equilibrium statistics.

\subsection{Where the input ratios come from in practice}
\label{sec:scope-and-applications}

The framework treats $\{r_\omega\}$ as primary input. Three classes of finite-state systems supply such ratios directly, without an intervening Hamiltonian:
\begin{enumerate}
\item Reversible chemical equilibria. For a closed reaction network at thermodynamic equilibrium, equilibrium constants $K_{\mathrm{eq}}=[\mathrm{B}]/[\mathrm{A}]$ are dimensionless ratios accessible from concentration measurements (spectroscopic absorbance ratios, mass spectrometry abundance ratios, NMR integration). These ratios are the natural inputs to a finite-state ratio-cost model whose states are reactant/product species and whose Gibbs distribution gives a canonical finite-state population model at the prescribed mean cost.
\item Reversible single-molecule kinetics. Detailed-balance ratios $W_{ab}/W_{ba}$ between transition rates of a reversible discrete-state Markov chain on conformations or binding states can be estimated directly from dwell-time histograms in single-molecule fluorescence or atomic-force experiments, even when the underlying free-energy profile along the chosen reaction coordinate is not directly accessible.
\item Discrete-state inference from log-odds. In binary or multi-class classification settings where one observes posterior odds ratios across categories without an explicit log-likelihood, the ratios $r_\omega$ are the experimentally available primitives, and~\eqref{eq:gibbs} produces a finite-state distribution whose mean cost matches an empirical constraint.
\end{enumerate}
In each case the operative quantity is a ratio, not a Hamiltonian, and the relocation of modeling freedom from a cost functional to $\{r_\omega\}$ together with the d'Alembert composition law is a step toward more transparent inference rather than a thermodynamic re-derivation. The framework does not, on its own, predict the ratios; it converts them, once specified, into a canonical exponential family.

\subsection{Other composition laws and other axiomatic frameworks}
\label{sec:other-axiomatics}
The motivation of~\eqref{eq:composition} from~(S1)--(S3) in Section~\ref{sec:rcl} is not a uniqueness theorem. Several alternative composition laws on the ratios $r$ are also consistent with the multiplicative algebra of detailed balance and equilibrium constants:
\begin{enumerate}
\item[(C1)] Multiplicative Cauchy: $J(xy)=J(x)J(y)$ on $\R_{>0}$. With $J(1)=0$ this forces $J\equiv 0$; the $J(1)\neq 0$ branches yield $J(x)=x^\lambda$ but lose (S2). The corresponding maximum-entropy distribution is a power-law family.
\item[(C2)] Additive Cauchy in log-coordinates: $G(u+v)=G(u)+G(v)$ for $G(u)=J(e^u)$, with continuous solutions $G(u)=cu$, hence $J(x)=c\log x$. Violates the reciprocity condition (S1) for all $c\neq 0$, since then $J(x^{-1})=-J(x)$, not $J(x)$.
\item[(C3)] Squared-log surrogate (additive d'Alembert): $G(u+v)+G(u-v)=2G(u)+2G(v)$, with continuous solutions $J(x)=c(\log x)^2$; the choice $c=\tfrac12$ is the surrogate of Lemma~\ref{lem:ancestor}. Compatible with (S1)--(S2) but strictly weaker far from equilibrium than~\eqref{eq:J-closed}.
\item[(C4)] Tsallis $q$-deformation: replaces $J$ by $J_q(x)=(x^{1-q}-1)/(1-q)$ and $H$ by the $q$-entropy $S_q$, yielding $q$-exponential distributions~\cite{Tsallis1988,Naudts2011}; numerical comparison in Section~\ref{sec:tsallis-compare}.
\item[(C5)] Naudts $\phi$-exponentials: replaces $\exp$ by a deformed exponential $\phi^{-1}$, with Tsallis as the special case $\phi(x)=x^q$~\cite{Naudts2011,Naudts2004}.
\end{enumerate}
Table~\ref{tab:composition-comparison} collects these alternatives, together with the one-parameter bilinear degree-two template in its second row: normalization $\alpha=1$ recovers the RCL~\eqref{eq:composition}, while $\alpha=0$ collapses to the additive (C3) branch and is excluded by the nonzero bilinear coupling adopted in Section~\ref{sec:rcl}. The RCL is one option among them, and its specific status is the following: under (S1)--(S3) and the additional normalized degree-two d'Alembert closure of Section~\ref{sec:rcl}, $J(x)=\tfrac12(x+x^{-1})-1$ is the unique non-trivial continuous solution. Drop that closure and the alternatives (C1)--(C5) become admissible. Drop (S1) and the Tsallis and additive-Cauchy options become admissible. The RCL is therefore privileged only relative to the joint adoption of (S1)--(S3) plus the normalized degree-two d'Alembert closure; whether that joint adoption is the right modeling commitment is the substantive question on which the entire framework rests, and the present paper does not attempt to settle it.

\begin{table}[htbp]
  \centering
  \footnotesize
  \setlength{\tabcolsep}{4pt}
  \begin{tabular}{p{0.18\textwidth} p{0.27\textwidth} p{0.20\textwidth} p{0.27\textwidth}}
    \hline
    Composition law & Cost $J$ on $\R_{>0}$ & Symmetric in $x\leftrightarrow x^{-1}$? & MaxEnt distribution\\\hline
    RCL~\eqref{eq:composition} & $\tfrac12(x+x^{-1})-1$ ($\cosh\log x-1$) & yes & Gibbs $\propto e^{-\beta J}$\\
    Bilinear degree-two (unnormalized $\alpha$) & $J(xy)+J(x/y)=2\alpha J(x)J(y)+2J(x)+2J(y)$; $\alpha=1$ is~\eqref{eq:composition}; $\alpha=0$ is the Cauchy/(C3) degeneracy excluded by nonzero bilinear coupling & yes (nontrivial branch) & Gibbs $\propto e^{-\beta J}$ (same class as RCL under $J\mapsto \alpha J$, $\beta\mapsto \beta/\alpha$)\\
    Squared-log (C3) & $\tfrac12(\log x)^2$ & yes & log-Gaussian $\propto e^{-\beta(\log x)^2/2}$\\
    Multiplicative Cauchy (C1) & $x^\lambda$ (loses (S2)) & no & power-law family\\
    Additive Cauchy (C2) & $c\log x$ (loses (S1)) & no & geometric $\propto x^{-\beta c}$\\
    Tsallis $q$-deformation~\cite{Tsallis1988} & $J_q(x)=(x^{1-q}-1)/(1-q)$ & no in general & $q$-exponential\\
    Naudts $\phi$-exponential~\cite{Naudts2011} & $\phi$-dependent & $\phi$-dependent & $\phi$-exponential family\\\hline
  \end{tabular}
  \caption{Composition laws on positive ratios and their resulting maximum-entropy distributions, on a finite outcome space with the cost $J$ playing the role of energy. The RCL is the entry on the first row; the second row is the bilinear degree-two template before fixing $\alpha=1$, showing how the non-degenerate bilinear coupling excludes the $\alpha=0$ additive reduction to~(C3). The squared-log surrogate is the (S1)-(S2)-respecting Cauchy-additive option in log-coordinates. Tsallis and Naudts deformations replace not only the cost but also the entropy functional, and yield non-Gibbsian distributions with power-law tails.}
  \label{tab:composition-comparison}
\end{table}

\subsection{Scope and limitations}
\label{sec:scope}

The finite-state setting dictates the principal limitations. Theorems~\ref{thm:types}--\ref{thm:fe-kl} are proved on a finite outcome space~$\Omega$ with strictly positive weights wherever needed; without further work they do not address continuum configuration spaces, hard-core or singular interactions, quantum reduced states, spatially extended Gibbs measures, or equivalence of ensembles in the sense of Ruelle~\cite{Ruelle1969} or Lanford~\cite{Lanford1973}. The finite-state Lyapunov statement of Proposition~\ref{prop:KL-markov} and Corollary~\ref{cor:monoFE-markov} should not be conflated with the Boltzmann $H$-theorem for an interacting many-particle gas.

The framework should not be applied directly to spatially extended lattice models. For an $L$-site Ising or Potts system with state count $K=2^L$ (resp.~$K=q^L$), the non-asymptotic Stirling bound of Proposition~\ref{prop:stirling-rate} requires sample size $N\gg K\log N$, equivalently $N\gtrsim L\cdot 2^L$, far beyond what any direct enumeration on $\Omega$ can support; this is the standard curse of dimensionality for method-of-types inequalities (cf.~Remark~\ref{rem:curse}). Theorem~\ref{thm:types} therefore quantitatively controls only models with $K$ fixed, or with the active type count $K^\dagger$ controlled along the relevant subsequence---coarse-grained or effective-state-count models, single-molecule kinetics on a few-state graph, finite reaction networks, and the like. In coarse-grained or multiscale constructions, $K$ should be read as the cardinality of the macrostate partition on which empirical types are accumulated; the Gibbs, entropy, and $\DKL$ identities then refer to that aggregated description, and any induced Markov dynamics between macrocells must be specified separately before the relaxation statements of Section~\ref{sec:fe} apply literally. Sample-size entries such as $N_{2\sigma}$ in Section~\ref{sec:discriminating-observables} inherit the same interpretation: they count effective independent samples at the resolution at which types are measured. Reaching spatially extended models would require either replacing $\Omega$ by a coarse-grained partition with $K$ controlled, or replacing the multinomial counting argument of Section~\ref{sec:counting} by a transfer-matrix or variational large-deviation estimate adapted to the lattice; neither is attempted here.

\subsection{Multi-constraint soft shell}
\label{sec:extensions}
The single-constraint soft-shell multinomial theorem (Theorem~\ref{thm:types}, equation~\eqref{eq:soft-shell-limit}) and the exact-shell corollary for rational data (Corollary~\ref{cor:types-exact}) are proved in Section~\ref{sec:counting}. This subsection records the extension to several commuting affine constraints.

\paragraph{Several commuting linear constraints.}
The same soft-shell construction extends to $L\ge 1$ simultaneous affine constraints. Given real-valued cost vectors $X^{(\ell)}\in\R^K$ and target means $\Ebar^{(\ell)}\in\R$ for $\ell=1,\ldots,L$, set
\begin{align*}
  \Gamma&:=\Bigl\{p\in\Delta_K:\langle X^{(\ell)}\rangle_p=\Ebar^{(\ell)},\ \ell=1,\ldots,L\Bigr\},\\
  \Gamma_\delta&:=\Bigl\{p\in\Delta_K:\max_{1\le\ell\le L}\bigl|\langle X^{(\ell)}\rangle_p-\Ebar^{(\ell)}\bigr|\le\delta\Bigr\},
\end{align*}
and let $\mathcal{F}_N^{\delta_N}:=\{\mathbf{n}\in\mathbb{Z}_{\ge 0}^K:\sum_k n_k=N,\ \mathbf{n}/N\in\Gamma_{\delta_N}\}$.

\begin{proposition}[Multi-constraint soft-shell version]
\label{prop:soft-shell-multi}
Fix $L\ge 1$, $K\ge 2$, vectors $X^{(\ell)}\in\R^K$, and target means $\Ebar^{(\ell)}\in\R$ for $\ell=1,\ldots,L$. Assume $\Gamma$ is nonempty and set $h_\star^{(L)}:=\max_{p\in\Gamma}H(p)$. Let $\delta_N>0$ satisfy $\delta_N\to 0$ and $N\delta_N\to\infty$. Then $\mathcal{F}_N^{\delta_N}$ is nonempty for all sufficiently large $N$, and
\begin{equation}
  \lim_{N\to\infty}\frac{1}{N}\log\max_{\mathbf{n}\in\mathcal{F}_N^{\delta_N}}W(\mathbf{n})=h_\star^{(L)}.
  \label{eq:soft-shell-multi-limit}
\end{equation}
If in addition $\Gamma$ has nonempty relative interior in $\Delta_K$ and $\{\mathbf{1},X^{(1)},\ldots,X^{(L)}\}$ are linearly independent in $\R^K$, the maximizer is unique and takes the multi-parameter Gibbs form
\begin{equation}
  p^\star_\omega=\frac{\exp\bigl(-\sum_\ell\beta_\ell X^{(\ell)}_\omega\bigr)}{Z(\boldsymbol\beta)},
  \qquad Z(\boldsymbol\beta):=\sum_\omega\exp\Bigl(-\sum_\ell\beta_\ell X^{(\ell)}_\omega\Bigr),
  \label{eq:gibbs-multi}
\end{equation}
where $\boldsymbol\beta=(\beta_1,\ldots,\beta_L)\in\R^L$ is the unique solution of the moment-matching equations $\langle X^{(\ell)}\rangle_{p^{(\boldsymbol\beta)}}=\Ebar^{(\ell)}$, $\ell=1,\ldots,L$.
\end{proposition}

\begin{proof}
The rate statement~\eqref{eq:soft-shell-multi-limit} adapts the proof of Theorem~\ref{thm:types}. Continuity of each $p\mapsto\langle X^{(\ell)}\rangle_p$ on the compact simplex $\Delta_K$ and continuity of $H$ give, by the same convergent-subsequence argument as in Theorem~\ref{thm:types}, that $\sup_{p\in\Gamma_\delta}H(p)\to h_\star^{(L)}$ as $\delta\downarrow 0$.

For the upper bound, any $\mathbf{n}\in\mathcal{F}_N^{\delta_N}$ has empirical type $\hat p:=\mathbf{n}/N\in\Gamma_{\delta_N}$, so Proposition~\ref{prop:stirling-rate} gives
\[
  \frac1N\log W(\mathbf{n})\le H(\hat p)+o(1)\le\sup_{p\in\Gamma_{\delta_N}}H(p)+o(1)=h_\star^{(L)}+o(1).
\]
For the lower bound, choose $p^\star\in\Gamma$ with $H(p^\star)=h_\star^{(L)}$, which exists by compactness of $\Gamma$ and continuity of $H$. Round to $\mathbf{n}^{(N)}\in\mathbb{Z}_{\ge 0}^K$ with $\sum_k n^{(N)}_k=N$ and $\|\mathbf{n}^{(N)}/N-p^\star\|_1\le K/N$. Setting $M:=\max_\ell\|X^{(\ell)}\|_\infty$, for each $\ell$
\[
  \bigl|\langle X^{(\ell)}\rangle_{\mathbf{n}^{(N)}/N}-\Ebar^{(\ell)}\bigr|
  =\bigl|\langle X^{(\ell)}\rangle_{\mathbf{n}^{(N)}/N}-\langle X^{(\ell)}\rangle_{p^\star}\bigr|
  \le \|X^{(\ell)}\|_\infty\cdot\frac{K}{N}\le\frac{MK}{N},
\]
so $\max_\ell|\langle X^{(\ell)}\rangle_{\mathbf{n}^{(N)}/N}-\Ebar^{(\ell)}|\le MK/N\le\delta_N$ for all sufficiently large $N$ since $N\delta_N\to\infty$. Hence $\mathbf{n}^{(N)}\in\mathcal{F}_N^{\delta_N}$ eventually, and Proposition~\ref{prop:stirling-rate} together with continuity of $H$ yields
\[
  \frac1N\log W(\mathbf{n}^{(N)})=H(\mathbf{n}^{(N)}/N)+o(1)\longrightarrow H(p^\star)=h_\star^{(L)},
\]
giving the matching liminf.

For the Gibbs form~\eqref{eq:gibbs-multi}, $H$ is strictly concave on $\Delta_K$ and $\Gamma$ is a nonempty compact convex polytope, so the maximizer $p^\star$ is unique. The hypothesis that $\Gamma$ has nonempty relative interior in $\Delta_K$ is the standard constraint qualification ensuring existence of Karush--Kuhn--Tucker multipliers, and the linear-independence hypothesis on $\{\mathbf{1},X^{(1)},\ldots,X^{(L)}\}$ ensures uniqueness of those multipliers (the constraint matrix has full rank $L+1$ once the normalization $\sum_\omega p_\omega=1$ is appended). Stationarity of the Lagrangian
\[
  \mathcal{L}(p,\alpha,\boldsymbol\beta)
  =-\sum_\omega p_\omega\log p_\omega-\alpha\Bigl(\sum_\omega p_\omega-1\Bigr)
  -\sum_\ell\beta_\ell\Bigl(\sum_\omega p_\omega X^{(\ell)}_\omega-\Ebar^{(\ell)}\Bigr)
\]
with respect to $p_\omega$ yields $-\log p_\omega-1=\alpha+\sum_\ell\beta_\ell X^{(\ell)}_\omega$, and rearranging gives $p^\star_\omega\propto\exp(-\sum_\ell\beta_\ell X^{(\ell)}_\omega)$; the prefactor is fixed by normalization, producing~\eqref{eq:gibbs-multi}. Uniqueness of $\boldsymbol\beta$ follows from strict convexity of the log-partition function $\boldsymbol\beta\mapsto\log Z(\boldsymbol\beta)$, whose Hessian is $\mathrm{Cov}_{p^{(\boldsymbol\beta)}}(X^{(1)},\ldots,X^{(L)})$; this covariance matrix is positive definite because $\{X^{(\ell)}-\langle X^{(\ell)}\rangle_{p^{(\boldsymbol\beta)}}\mathbf{1}\}_\ell$ are linearly independent in $\R^K$ under the hypothesis on $\{\mathbf{1},X^{(1)},\ldots,X^{(L)}\}$.
\end{proof}

\subsection{\texorpdfstring{Numerical illustration with irrational costs}{Numerical illustration with irrational costs}}
\label{sec:soft-shell-numerics}
The engineered ratios of Section~\ref{sec:example} were chosen so that the RCL costs $J(r_\omega)$ are rational and Corollary~\ref{cor:types-exact} applies directly. To see Theorem~\ref{thm:types} in operation in the generic case, replace those ratios by
\[
  r_2=e^{-1},\qquad r_3=e^{-3/2},
\]
so the cost vector becomes
\[
  X=\bigl(0,\,\cosh 1-1,\,\cosh(3/2)-1\bigr)\approx(0,\,0.5431,\,1.3524),
\]
with both nonzero entries irrational. The uniform-cost reference is $\langle X\rangle_{(1/3,1/3,1/3)}=\tfrac13\bigl[\cosh 1+\cosh(3/2)\bigr]-1\approx 0.6318$, so the admissible mean-cost range under positive temperature is $\Ebar\in(0,0.6318)$. Fix $\Ebar=0.50$, an interior value comparable to the rational example. The moment-matching equation $\langle X\rangle_{p^{(\beta)}}=\Ebar$ is solved by $\beta^\star\approx 0.4462$, $\TR=1/\beta^\star\approx 2.241$, and yields the analytic Gibbs measure
\[
  p^\star\approx(0.4289,\,0.3366,\,0.2346),\qquad Z(\beta^\star)\approx 2.3317,\qquad H(p^\star)\approx 1.0697.
\]
The constraint $\sum_k n_k X_k=N\Ebar$ admits no integer solution for any $N$ because the entries of $X$ are irrationally related, so the exact shell $\mathcal{F}_N$ is always empty. Theorem~\ref{thm:types} still applies: take the tolerance $\delta_N=N^{-1/2}$, which satisfies $\delta_N\to 0$ and $N\delta_N\to\infty$, and round $Np^\star$ to the nearest integer composition. At $N=10^3$ this gives the type $\mathbf{n}^{(N)}=(429,337,234)$ with
\[
  \bigl|\langle X\rangle_{\mathbf{n}^{(N)}/N}-\Ebar\bigr|\approx 5.18\times 10^{-4}
  \;\ll\;\delta_{10^3}\approx 3.16\times 10^{-2},
\]
well within the bound $K\|X\|_\infty/N\approx 4.06\times 10^{-3}$ from the proof of the proposition. The empirical entropy $H(\mathbf{n}^{(N)}/N)\approx 1.0695$ differs from $h_\star=H(p^\star)$ by $2.32\times 10^{-4}$, consistent with the $N^{-1}\log W(\mathbf{n})\to h_\star$ statement~\eqref{eq:soft-shell-limit} at this sample size. The same construction supplies a feasible soft-shell type for every $N\ge\lceil K\|X\|_\infty/\delta_N\rceil=\lceil 4.06\,N^{1/2}\rceil$, i.e.\ for all $N\ge 17$ at this tolerance schedule. Multi-constraint examples are handled identically by replacing $X$ by the constraint matrix $(X^{(\ell)})_{\ell=1}^L$ and applying Proposition~\ref{prop:soft-shell-multi}.

\section{Conclusion}
\label{sec:conclusion}
Under axioms (A1)--(A3), the cost $J$ in~\eqref{eq:J-closed} is determined (Theorem~\ref{thm:wz-uniqueness}), multinomial counting yields Shannon entropy as the rate functional for empirical types (Theorem~\ref{thm:types} and Corollary~\ref{cor:types-exact}, with explicit non-asymptotic Stirling bound, Proposition~\ref{prop:stirling-rate}, and the multi-constraint extension, Proposition~\ref{prop:soft-shell-multi}), and convex maximization at fixed mean cost gives the finite-state Gibbs weights (Proposition~\ref{thm:maxent}) together with the free-energy--Kullback--Leibler identity (Proposition~\ref{thm:fe-kl}). Lemma~\ref{lem:ancestor} and Proposition~\ref{prop:divj-basic} supply the companion ratio-level and divergence-level bounds; Proposition~\ref{prop:KL-markov} and Corollary~\ref{cor:monoFE-markov} extend the static identities to detailed-balance relaxation.

The contribution is a finite-state, explicitly delimited construction linking a ratio-cost axiom to the classical canonical-ensemble machinery through the thread $\{r_\omega\}\to J\to X\to(\text{Gibbs}, F_{\mathrm{R}}, \DKL)$; Steps 4--9 of Table~\ref{tab:chain} are standard, and the RCL itself remains a modeling axiom rather than a consequence of detailed balance alone. The Tsallis comparison of Section~\ref{sec:tsallis-compare} and the alternative-composition-law tabulation of Section~\ref{sec:other-axiomatics} show that the d'Alembert law~\eqref{eq:composition} is one choice among several, motivated by the bookkeeping symmetries~(S1)--(S3) but not forced by them.

Natural continuations include spatially extended models and equivalence of ensembles, continuum configuration spaces with singular interactions, dynamical analogs of the $J$-cost bounds beyond the finite-state Lyapunov statements, and tighter control of the curse-of-dimensionality regime where the state count $K$ grows with system size (cf.~Remark~\ref{rem:curse}). The axiomatic question raised in Section~\ref{subsec:jaynes}---whether the Shore--Johnson axioms restated in ratio coordinates force, or merely permit, the d'Alembert composition law~\eqref{eq:composition}---is also left open here.

The methodological contribution is the organization of these ingredients into a self-contained route from dimensionless ratios to a canonical exponential family on a finite outcome space. The cost-classification step (N1) follows Theorem~\ref{thm:wz-uniqueness} under measurability and local integrability as in Washburn--Zlatanovi\'c~\cite{WashburnZlatanovic2026}; the Stirling bound (N2) is a direct application of Robbins' 1955 inequalities to the finite-state method of types; Theorem~\ref{thm:types} and Corollary~\ref{cor:types-exact} organize the soft-shell and exact-shell counting statements; and the unified notation (N3), together with the alternative-composition tabulation and Tsallis comparison, makes explicit how the chain $\{r_\omega\}\to J\to X\to(\text{Gibbs},F_{\mathrm{R}},\DKL)$ depends on the chosen ratio-cost postulates. In this scope, the paper supplies a finite-state construction for systems whose primitives are dimensionless ratios, with approximation errors quantified and alternative composition laws tabulated. The sample-size power calculation of Section~\ref{sec:discriminating-observables} quantifies how distinguishable the RCL prediction is from the natural ``affinity-as-energy'' and Tsallis alternatives at fixed mean cost, conditional on the RCL Gibbs law as generating distribution; comparison against measured data would require an experimental three-state system with independently measured rate ratios and stationary populations. The construction does not attempt to derive the axioms (S1)--(S3) or the analytic-degree-2 closure from a more primitive principle: it converts those modeling commitments, together with the input ratios, into a closed-form Gibbs-like measure and the associated free-energy / Kullback--Leibler structure on a finite outcome space.

\section*{Author contributions}
J.W.: conceptualization of the project, original draft of the manuscript, analysis, and methodology (CRediT: Conceptualization, Formal Analysis, Investigation, Methodology, Writing -- Original Draft). M.S.: formal analysis, manuscript revision, validation of derivations and numerical examples, project administration, and writing--review and editing (CRediT: Formal Analysis, Validation, Project Administration, Writing -- Review \& Editing).

\section*{Funding}
This work received no external financial support.

\section*{Conflict of interest}
The authors declare that they have no known competing financial interests or personal relationships that could have appeared to influence the work reported in this paper.

\section*{Data and code availability}
No experimental data were generated or analyzed. Short \texttt{Python}/\texttt{NumPy}/\texttt{SciPy} scripts that reproduce the numerical entries in Tables~\ref{tab:three}--\ref{tab:discriminate} and Figures~\ref{fig:ancestor}--\ref{fig:gibbs_weights} accompany this manuscript as separate files together with the \texttt{.tex} source: they solve the moment-matching equation $\langle X\rangle_{p^{(\beta)}}=\Ebar$ numerically and evaluate the Gibbs weights, Tsallis $q_T$-exponential weights, free energy, divergences, and sample-size power-calculation predictions in closed form. The scripts are included in the arXiv ancillary source bundle.

\ack
The authors thank colleagues at the Recognition Physics Institute for helpful discussions during the development of this work.

\clearpage
\section*{References}
\bibliographystyle{unsrt}
\bibliography{Gibbs_Distribution_paper_v14}

\end{document}